\documentclass[11pt,a4paper]{article}

\usepackage{a4wide,url}
\usepackage{graphicx,amssymb,amsmath,amsthm,color}

\graphicspath{{figures/}{}}

\newtheorem{definition}{Definition}
\newtheorem{lemma}[definition]{Lemma}

\newtheorem{theorem}[definition]{Theorem}
\newtheorem{corollary}[definition]{Corollary}

\def\Reals{\ensuremath{\mathbb{R}}}

\def\ZZ{\ensuremath{\mathbb{Z}}}

\def\A{\ensuremath{\mathcal{A}}}
\def\I{\ensuremath{\mathbb{G}}}
\def\cycle{\ensuremath{\tau}}
\def\true{{\rm T}}
\def\false{{\rm F}}

\newcommand\segment[1]{\ensuremath{\overline{#1}}}

\DeclareMathOperator{\lengthBIS}{len}
\newcommand\length{\lengthBIS_{\I}}

\DeclareMathOperator{\reverse}{reverse}
\DeclareMathOperator{\lca}{lca}
\DeclareMathOperator{\polylog}{polylog}

\begin{document}
\pagestyle{plain}

\title{Minimum cell connection and separation\\ in line segment arrangements\thanks{Part of the results contained here were presented at EuroCG 2011~\cite{acgk-11}. Research partially conducted at the 9th McGill - INRIA Barbados Workshop on Computational Geometry, 2010.}}

\author{Helmut Alt%
            \thanks{Institut f{\"ur} Informatik, Freie Universit{\"a}t Berlin,
    Takustra{\ss}e 9, D-14195 Berlin, Germany, {\tt \{alt, panos\}@mi.fu-berlin.de}}
    \footnote{Research supported by the German Science Foundation (DFG) under grant Kn~591/3-1.}
        \and
        Sergio Cabello\thanks{Department of Mathematics, IMFM and FMF, University of Ljubljana,
                Jadranska 19, SI-1000 Ljubljana, Slovenia, {\tt sergio.cabello@fmf.uni-lj.si.}}
       \footnote{Research was partially supported by the Slovenian Research Agency, program P1-0297.}
        \and
        Panos Giannopoulos\footnotemark[2]{\ }\footnotemark[3]
        \and
        Christian Knauer\thanks{Institut f{\"ur} Informatik, Universit{\"a}t Bayreuth
    Universit{\"a}tsstra{\ss}e 30
    D-95447 Bayreuth
    Germany, {\tt christian.knauer@uni-bayreuth.de}}{\ }\footnotemark[2]
}
\date{}
\maketitle

\begin{abstract}
We study the complexity of the following cell connection and separation problems in segment arrangements.
Given a set of straight-line segments in the plane and two points $a$ and $b$ in different cells of the induced arrangement:
\begin{itemize}
\setlength{\itemsep}{-\parsep}
\item[(i)] compute the minimum number of segments one needs to remove so that there is a path connecting $a$ to $b$ that does not intersect any of the remaining segments;
\item[(ii)] compute the minimum number of segments one needs to remove so that the arrangement induced by the remaining segments has a single cell;
\item[(iii)] compute the minimum number of segments one needs to retain so that any path connecting $a$ to $b$ intersects some of the retained segments.
\end{itemize}

We show that problems (i) and (ii) are NP-hard and discuss some special, tractable cases. Most notably, we provide a linear-time algorithm for a variant of problem (i) where the path connecting $a$ to $b$ must stay inside a given polygon $P$ with a constant number of holes, the segments are contained in $P$, and the endpoints of the segments are on the boundary of $P$.
For problem (iii) we provide a cubic-time algorithm.
\end{abstract}

\section{Introduction}

In this paper we study the complexity of some natural optimization problems in segment arrangements. Let $S$ be a set of straight-line segments in $\Reals^2$, $\mathcal{A}(S)$ be the arrangement induced by $S$, and $a, b$ be two points not incident to any segment of $S$ and in different cells of $\mathcal{A}(S)$.

In the {\sc $2$-Cells-Connection} problem we want to compute a set of segments $S'\subseteq S$ of minimum cardinality with the property that $a$ and $b$ belong to the same cell of $\mathcal{A}(S\setminus S')$. 
In other words, we want to compute 
an $a$-$b$ path that \emph{crosses} the minimum number of segments of $S$ counted without multiplicities.
The \emph{cost} of a path is the total number of segments it crosses.

In the {\sc All-Cells-Connection} problem we want to compute a set $S'\subseteq S$ of minimum cardinality such that $\mathcal{A}(S\setminus S')$ consists of one cell only.

In the {\sc $2$-Cells-Separation} problem we want to compute a set $S'\subseteq S$ of minimum cardinality that \emph{separates} $a$ and $b$, i.e., $a$ and $b$ belong to different cells of $\mathcal{A}(S')$ -- equivalently -- any $a$-$b$ path intersects some segment of $S'$. 

Apart from being interesting in their own right, the problems we consider here are also natural abstractions of problems concerning sensor networks. Each segment is surveyed (covered) by a sensor, and the task is to find the minimum number of sensors of a given network over some domain that must be switched on or off so that: an intruder can be detected when walking between two given points ({\sc $2$-Cells-Separation}), or can walk freely between two given points ({\sc $2$-Cells-Connection}) or can reach freely any point ({\sc All-Cells-Connection}). Because of these applications, it is worth considering a variant where the segments lie inside a given polygon $P$ with holes and have their endpoints on the boundary of $P$, and the $a$-$b$ path must also stay inside $P$. See Fig.~\ref{fig:polygon} for an example of this last scenario. We refer to these variants as the restricted {\sc $2$-Cells-Separation} or {\sc $2$-Cells-Connection} in a polygon.

\medskip
\noindent
{\bf Our results.} We provide an algorithm that solves {\sc $2$-Cells-Separation} in $\mathcal{O}(n^2+nk)$ time, where $k$ is the number of pairs of segments that intersect. The same algorithm, with an extra logarithmic factor, works for a generalization where the segments are weighted. The algorithm itself is simple, but its correctness is not at all obvious. We justify its correctness by considering an appropriate set of cycles in the intersection graph and showing that it satisfies the so-called \emph{3-path condition}~\cite{t-egsnc-90} (see also \cite[Chapter 4]{mt-gs-01}). The use of the 3-path condition for solving {\sc $2$-Cells-Separation} is surprising and makes the connection to topology clear. 

We show that both {\sc $2$-Cells-Connection} and {\sc All-Cells-Connection} are NP-hard even when the segments are in general position. The first result is given by a careful reduction from {\sc Max-$2$-Sat}, which also implies APX-hardness. The second one follows from a straightforward reduction that uses a connection to the feedback vertex set problem in the intersection graph of the segments and holds even if there are no proper segment crossings. Also, when any three segments may intersect only at a common endpoint, {\sc $2$-Cells-Connection} is fixed-parameter tractable with respect to the number of proper segment crossings. 

Finally, we consider the restricted problems in a polygon. The restricted {\sc $2$-Cells-Separation} in a polygon is easily reduced to the general weighted version and thus can be solved efficiently. The restricted {\sc $2$-Cells-Connection} in a polygon remains NP-hard but can be solved in near-linear time for any fixed number of holes. The approach for this latter result uses homotopies to group the segments into clusters with the property that any cluster is either contained or disjoint from the optimal solution. 

\begin{figure}[t]
\centering
\includegraphics[width=0.47\textwidth]{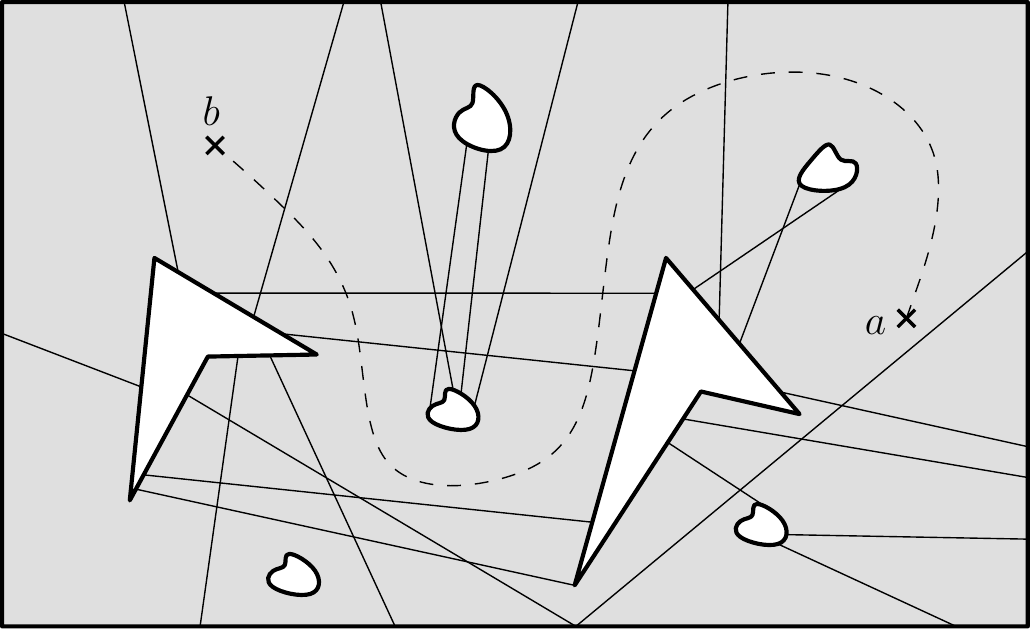}
\caption{A polygon with holes and a minimum-cost $a$-$b$ path.}
\label{fig:polygon}
\end{figure}

\medskip
\noindent
{\bf Related work.} Our NP-hardness proof for {\sc $2$-Cells-Connection} has been carefully extended by Kirkpatrick and Tseng~\cite{tseng-thesis}, who showed that the {\sc $2$-Cells-Connection} remains NP-hard even for \emph{unit-length} segments. However, their result does not imply APX-hardness for unit-length segments. 
The related problem of finding (from scratch) a set of segments with minimum total length that forms a barrier between two specified regions in a polygonal domain has been shown to be polynomial-time solvable by Kloder and Hutchinson~\cite{KH07}.

The problems we consider can of course be considered for other geometric objects, most notably unit disks. To this end, closely related work was done by Bereg and Kirkpatrick~\cite{BK09}, who studied the counterpart of {\sc $2$-Cells-Connection} in arrangements of unit disks and gave a $3$-approximation algorithm.  While the complexity of {\sc $2$-Cells-Connection} for unit (or arbitrary) disks is still unknown, there exist polynomial-time algorithms for restricted belt-shaped and simple polygonal domains~\cite{KLA07}. Simultaneously and independently to our work, Gibson et al.~\cite{gkv-ipud-11} have considered the problem of separating $k$ points in an arrangements of disks and provided a polynomial-time $O(1)$-approximation algorithm. Their approach is based on building a solution by considering several instances {\sc $2$-Cells-Separation} on arrangements of disks, which they can solve approximately.

\section{Separating two cells}

In this section we provide a polynomial-time algorithm for {\sc $2$-Cells-Separation}.
We will actually solve a weighted version, where we have a weight function $w$ assigning 
weight $w(s)\ge 0$ to each segment $s\in S$. For any subset $S'\subseteq S$ we define its weight
$w(S')$ as the sum of the weights over all segments $s\in S'$. 
The task is to find a minimum weight subset $S'\subseteq S$ that separates two given points $a$ and $b$.
Our time bounds will be expressed as a function of $n$, the number of segments in $S$,
and $k$, the number of pairs of segments in $S$ that intersect.
We first describe the algorithm, and then justify its correctness.

We assume for simplicity of exposition that the segment $\segment{ab}$ is vertical
and does not contain any endpoint of $S$ or any vertex of $\A(S)$.

Let $\gamma$ be a polygonal path contained in $\bigcup S$, possibly with self-intersections.
Because of our assumption on general position, no vertex
of $\gamma$ is on the segment $\segment{ab}$.
We define $N(\gamma;a,b)$ as the number of oriented intersections
of $\gamma$ with $\segment{ab}$: a crossing where $\gamma$ goes from the left to the right of $\segment{ab}$
contributes $+1$ to $N(\gamma;a,b)$,
while a right-to-left crossing contributes $-1$ to $N(\gamma;a,b)$.
We have $N(\reverse (\gamma);a,b)= -N(\gamma;a,b)$.

In a graph, we will use the term \emph{cycle} for a closed walk without repeated vertices.
A polygonal path is \emph{simple} if it does not have self-intersections.

\subsection{The algorithm}
\label{sec:algorithm}

From $S$ we construct its intersection graph 
$\I=(S,\{ ss'\mid s\cap s'\not = \emptyset \})$.
See Fig.~\ref{fig:intersectiongraph}(a)-(b). 
Note that $\I$ has $k$ edges.
To each edge $ss'$ of $\I$ we attach a weight (abstract length) $w(s)+w(s')$.
Any distance in $\I$ will refer to these edge weights.
For any walk $\pi$ in $\I$ we use $\length (\pi)$ for its length, that is, the sum of the weights on its edges
counted with multiplicity, and $S(\pi)=V(\pi)$ for the set of segments along $\pi$.
For any spanning tree $T$ in $\I$ and any edge $e\in E(\I)\setminus E(T)$,
let $\cycle(T,e)$ denote the cycle obtained by concatenating the edge $e$ with 
the path in $T$ connecting both endpoints of $e$. 
(The actual orientation of $\cycle(T,e)$ will not be relevant.)

\begin{figure}
\centering
	\includegraphics[width=\textwidth]{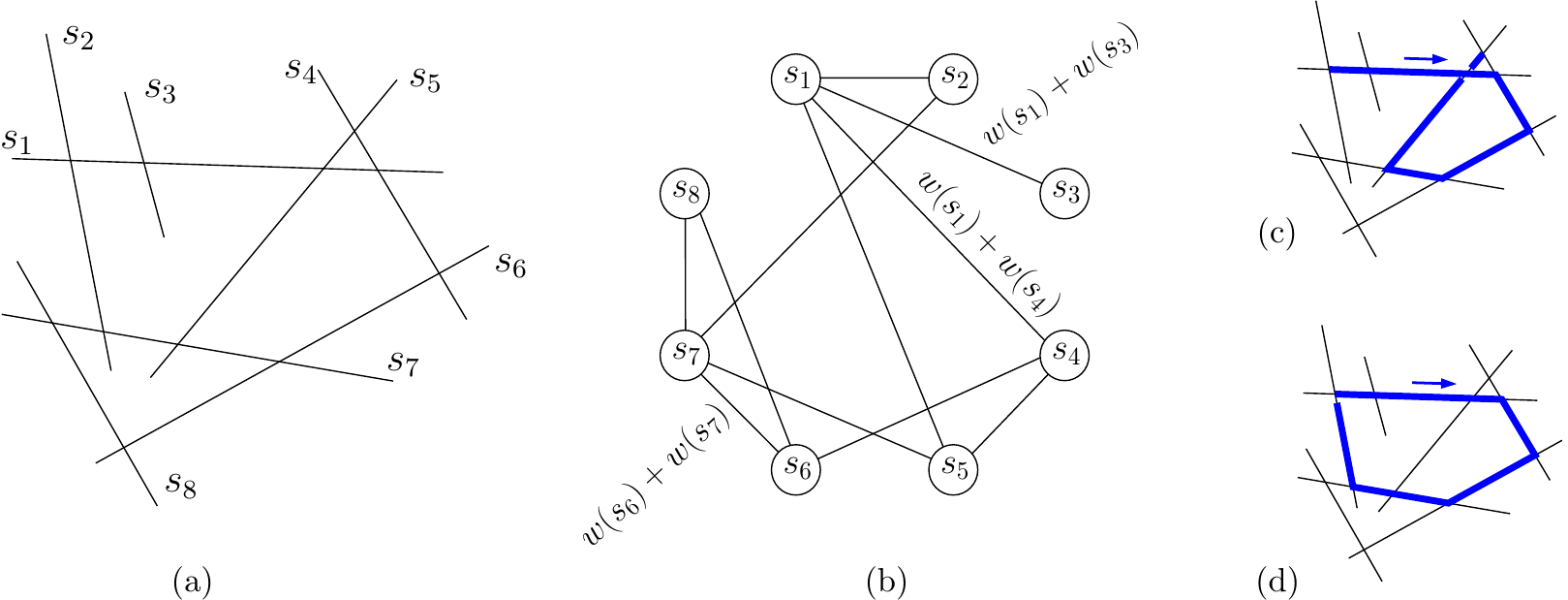}
	\caption{(a) A set of segments $S$. 
			(b) The corresponding intersection graph $\I$ with some of its edge-weights.
			(c) The polygonal path $\gamma(\pi)$ for the walk $\pi=s_2 s_1 s_4 s_6 s_7 s_5 s_4$.
			(d) The closed polygonal path $\gamma(\pi)$ for the closed walk $\pi=s_2 s_1 s_4 s_6 s_7 s_2$.}
	\label{fig:intersectiongraph}
\end{figure}	

Consider any walk $\pi=s_0 s_1\cdots s_t$ in $\I$.
This walk defines a polygonal path, denoted by $\gamma(\pi)$, 
which has vertices $x_0,x_1,\dots ,x_{t-1}$,
where $x_j=s_j\cap s_{j+1}$ for $j=0,\dots,t-1$. 
If $\pi$ is a closed walk with $s_0=s_t$, then we take $\gamma(\pi)$ to be 
a closed polygonal path whose last edge is $x_{t-1}x_0$, which is contained in $s_0$.
See Fig.~\ref{fig:intersectiongraph}(c)-(d). 
The polygonal path $\gamma(\pi)$ is contained in $\bigcup S(\pi)$.
Note that even if $\pi$ is a cycle
the closed polygonal path $\gamma(\pi)$ may have self-intersections.

For any segment $r\in S$, let $T_r$ be a shortest-path tree in $\I$ from $r$;
if there are several we fix one of them.
We will mainly use polygonal paths arising from cycles $\cycle(T_r,e)$,
where $e\in E(\I)\setminus E(T_r)$.
Thus we introduce the notation $\gamma(r,e)=\gamma (\cycle(T_r,e))$.
(Again, the actual orientation of $\gamma(r,e)$ will not be relevant.)

The algorithm is the following. We
compute the set
\[
	P=\{ (r,e)\in S\times E(\I) \mid \mbox{$e\in E(\I)\setminus E(T_r)$ and 
										$N(\gamma(r,e));a,b)\not = 0$}\},
\]
choose
\[
	(r^*,e^*)= \arg~ \min_{(r,e)\in P}~ \length(\cycle(T_r,e)),
\]
and return $S(\cycle(T_{r^*},e^*))$.
This finishes the description of the algorithm.
For analyzing it, it will be convenient to use the notation
$\tau^*=\cycle(T_{r^*},e^*)$ and $\gamma^*$ for the polygonal path $\gamma(\tau^*)$. 

The algorithm, as described above,
can be implemented in $\mathcal{O}(n^3k)$ time in a straightforward way.
We can speed up the procedure to obtain the following result.

\begin{lemma}\label{le:time}
	The pair $(r^*,e^*)$ can be computed in $\mathcal{O}(nk+ n^2\log n)$ time.
\end{lemma}
\begin{proof}
	The graph $\I$ can be constructed explicitly in $\mathcal{O}(n^2)$ time by checking
	each pair of segments, whether they cross or not.
	Recall that $\I$ has $k$ edges.

	For any segment $r\in S$, let us define
	\begin{align*}
		E_r ~~&=~~ \{ e\in E(\I)\setminus E(T_r)\mid (r,e)\in P\} \\
			&=~~  \{ e\in E(\I) \mid \mbox{$e\in E(\I)\setminus E(T_r)$ and 
										$N(\gamma(r,e));a,b)\not = 0$}\}. 
	\end{align*}
	Note that 
	\[
		P=\bigcup_{r\in S} \{ r\}\times E_r,
	\]
	and therefore 
	\[
		\min_{(r,e)\in P}~ \length(\cycle(T_r,e)) = 
		\min_{r\in S}~ \min_{e\in E_r}~ \length(\cycle(T_r,e)).
	\]
	Thus, $(r^*,e^*)$ can be computed by finding, for each $r\in S$,
	the value 
	\[
		\min_{e\in E_r}\length(\cycle(T_r,e)).
	\]
	We shall see that, for each fixed
	$r\in S$, such value  
	can be computed in $\mathcal{O}(k+n\log n)$ time. 
	It then follows that  $(r^*,e^*)$ can be found in 
	$|S|\times \mathcal{O}(k+n\log n)=\mathcal{O}(nk+ n^2\log n)$ time.	
	
	\begin{figure}
	\centering
		\includegraphics[width=\textwidth]{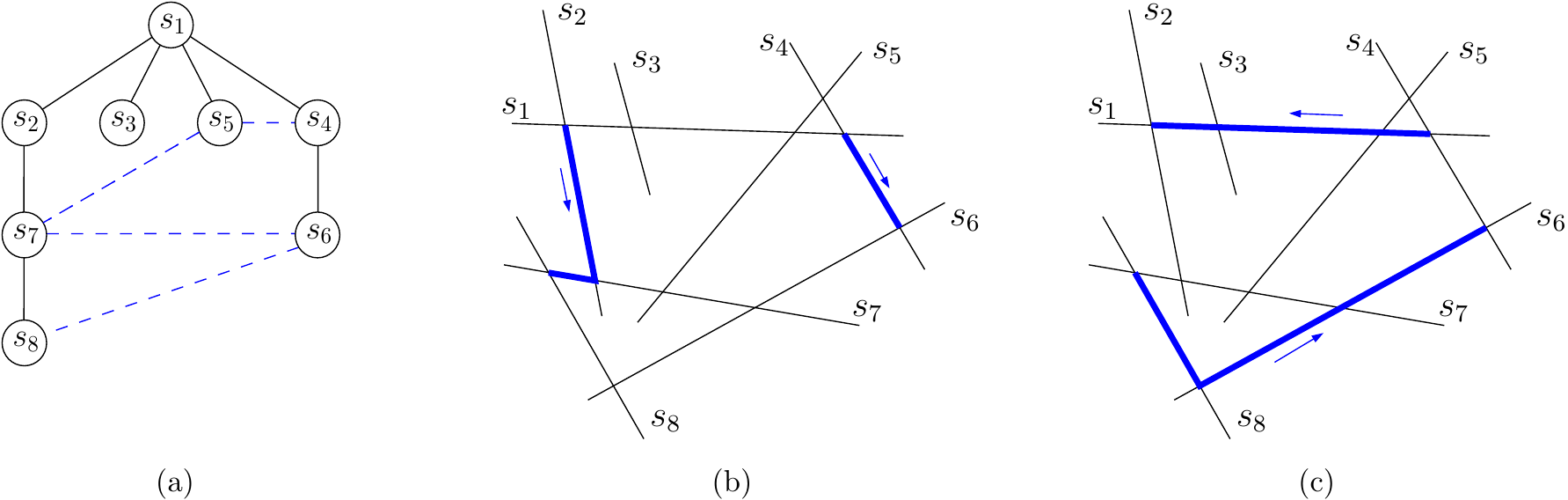}
		\caption{(a) Tree $T_{s_1}$ for the scenario of Fig.~\ref{fig:intersectiongraph} assuming unit weights in the segments. 
				In this case $C_{s_1}(s_8)=s_2$ and $C_{s_1}(s_6)=s_4$.
			(b) The polygonal paths $\gamma(T_{s_1}[s_8])$ and $\gamma(T_{s_1}[s_6])$.
			(c) The polygonal paths $\gamma(p_{s_1}(s_8)s_8s_6 p_{s_1}(s_6))= \gamma(s_7s_8s_6 s_4)$ 
				and $\gamma(C_{s_1}(s_6)s_1 C_{s_1}(s_8))= \gamma(s_4s_1s_2)$ that are used to compute $N(\gamma(s_1,s_6s_8);a,b)$
				in Lemma~\ref{le:time}.
			}
	\label{fig:intersectiongraph2}
	\end{figure}
	
	For the rest of the proof, let us fix a segment $r\in S$. 
	Computing $T_r$ takes $\mathcal{O}(|E(\I)| + |V(\I)|\log |V(\I)|) = \mathcal{O}(k+n\log n)$ time. 
	For any segment $s\in S$, $s\not =r$, let $T_r[s]$ denote the path in $T_r$ from $r$ 
	to $s$. We define $N_r(s)=N(\gamma(T_r[s]);a,b)$ 
	and define $C_r(s)$ to be the child of $r$ in the path $T_r[s]$.
	See~Fig.~\ref{fig:intersectiongraph2}(a)--(b).
	The values $N_r(s)$, $s\in S$, can be computed in $\mathcal{O}(n)$ 
	time using a BFS traversal of $T_r$: 
	if $p_r(s)$ is the parent of $s$ in $T_r$, we can compute $N_r(s)$
	from $N_r(p_r(s))$ in $\mathcal{O}(1)$ time using  
	\[ 
		N_r(s)~=~ N_r(p_r(s))+ N(\gamma(p_r(p_r(s))p_r(s)s) ;a,b).
	\]
	Similarly $C_r(s)$, $s\in S$, can be computed in $\mathcal{O}(n)$ time:
	we assign $C_r(s')=s'$ for all children $s'$ of $r$ and use
	that $C_r(s)=C_r(p_r(s))$ for any $s$ not adjacent to $r$.
	
	For $ss'\in E(\I)\setminus E(T_r)$ we have that $N(\gamma(r,ss'):a,b)$ is equal to
	\begin{align*}
		 N_r(s) + N(\gamma( p_r(s)ss' p_r(s'));a,b) - N_r(r,s') + N(\gamma( C_r(s')r C_r(s));a,b).
	\end{align*}
	See~Fig.~\ref{fig:intersectiongraph2}(b)--(c).
	(The negative sign comes from the reversal of $T_r[s]$.)
	Therefore, each $N(\gamma(r,ss');a,b)$ can be computed in $\mathcal{O}(1)$ time from the values $N_r(s)$, $N_r(s')$, $C_r(s')$, $C_r(s)$.
	It follows that $E_r$ can be constructed in $\mathcal{O}(|E(\I)|)=\mathcal{O}(k)$ time.
	
	The length of any cycle $\cycle(T_r,e)$ can be computed in $\mathcal{O}(1)$ time per cycle in a similar fashion.
	For each vertex $s$, we store at $s$ its distance $d_\I(r,s)$ from the root $r$.
	We also construct a data structure for finding lowest common ancestor ($\lca$) of two vertices in constant
	time. Such data structure can be constructed in $\mathcal{O}(n)$ time~\cite{lca1,lca2}. The length of a cycle can then be recovered using
	\[
		\length(\cycle (T_r,ss'))~=~ d_\I (r,s)+ w(s)+w(s') + d_\I(r,s') - 2 d_\I (r,\lca (s,s')).
	\]
	Equipped with this, we can in $\mathcal{O}(k)$ time compute
	\[
		\min_{e\in E_r}~ \length(\cycle(T_r,e)). \qedhere
	\]
\end{proof}

The following special case will be also relevant later on.
\begin{lemma}\label{le:time2}
	If the weights of the segments $S$ are $0$ or $1$, 
	then the pair $(r^*,e^*)$ can be computed in $\mathcal{O}(nk+ n^2)$ time.
\end{lemma}
\begin{proof}
	In this case, a shortest path tree $T_r$ can be computed 
	in $\mathcal{O}(|E(\I)| + |V(\I)|) = \mathcal{O}(k+n)$ time
	because the edge weights of $\I$ are $0$, $1$, or $2$.
	Using the approach described in the proof of Lemma~\ref{le:time} we spend $\mathcal{O}(k+n)$ per root $r\in S$,
	and thus spend $\mathcal{O}(nk+n^2)$ in total.
\end{proof}

\subsection{Correctness}

Consider the set of closed walks
\[
	\Pi = \{ \pi\mid \mbox{$\pi$ a closed walk in $\I$ with $N(\gamma(\pi);a,b)\not =0$} \}.
\] 
We have the following property, known as 3-path condition.

\begin{lemma}\label{le:3-path}
	Let $\alpha_0,\alpha_1,\alpha_2$ be 3 walks in $\I$ from $s$ to $s'$. For $i=0,1,2$,
	let $\pi_i$ be the closed walk obtained by concatenating $\alpha_{i-1}$ and the reverse of $\alpha_{i+1}$,
	where indices are modulo 3.
	If one of the walks $\pi_0,\pi_1,\pi_2$ is in $\Pi$, then at least two of them are in $\Pi$.
\end{lemma}
\begin{proof}
	\begin{figure}
	\centering
		\includegraphics[width=\textwidth]{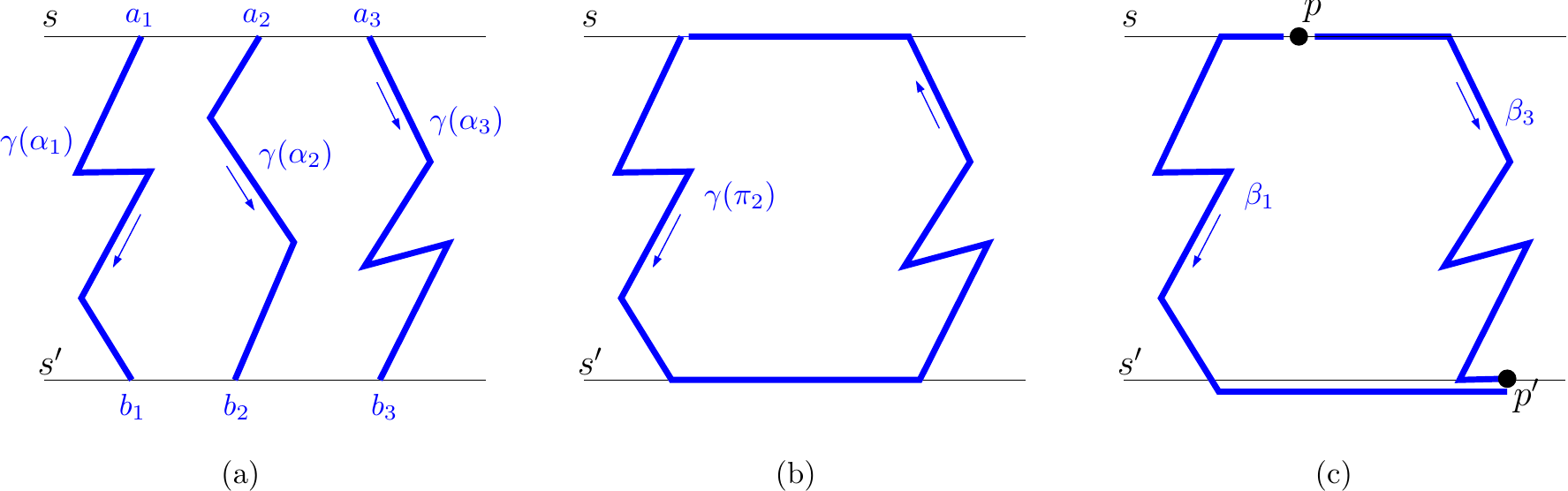}
		\caption{(a) Scenario in the proof of Lemma~\ref{le:3-path}.
			(b) The polygonal path $\gamma(\pi_2)$.
            (c) The polygonal paths $\beta_1$ and $\beta_3$. (The bottom of $\beta_1$ lies on $s'$. We draw it outside because of the common part with $\beta_3$.)}
	\label{fig:3pathcondition}
	\end{figure}
	(This result is a consequence of the group structure for relative $\ZZ_2$-homology. 
	We provide an elementary proof that avoids using homology.)
	For $i=0,1,2$, let $a_i$ be the starting vertex of the polygonal path $\gamma(\alpha_i)$
	and let $b_i$ be the ending vertex.
	The polygonal paths $\gamma(\alpha_0),\gamma(\alpha_1),\gamma(\alpha_2)$ 
	start on $s$ and finish on $s'$.
	However, they may have different endpoints.
	See Fig.~\ref{fig:3pathcondition}.
	To handle this, we choose a point $p$ on $s$ and a point $p'$ on $s'$,
	and define $\beta_i$ to be the polygonal path obtained
	by the concatenation of $pa_i$, $\gamma(\alpha_i)$, and $b_ip'$. 
	A simple but tedious calculation shows that, using indices modulo 3,
	\begin{align*}
		N(\gamma(\pi_i);a,b)~~=~~ N(\beta_{i-1};a,b)- N(\beta_{i+1};a,b).
	\end{align*}
	Indeed, since
	\[
		N(a_{i+1} a_{i-1};a,b)= N(p a_{i-1};a,b) + N(a_{i+1}p ;a,b)
	\]
	and 
	\[
		N(b_{i-1} b_{i+1};a,b)= N(p' b_{i+1};a,b) + N(b_{i-1}p';a,b),
	\]
	we have 
	\begin{align*}
		N(\gamma(\pi_i);a,b)~&=~~~~ N(\gamma(\alpha_{i-1}) ;a,b) 
							+ N(b_{i-1}b_{i+1};a,b) \\
							& ~~~~ + N(\reverse (\gamma(\alpha_{i+1}));a,b) 
							+ N(a_{i+1} a_{i-1};a,b)\\
							&=~~~~ N(\gamma(\alpha_{i-1}) ;a,b) 
							+ N(p' b_{i+1};a,b) + N(b_{i-1}p';a,b) \\
							& ~~~~ - N(\gamma(\alpha_{i+1});a,b) 
							+ N(p a_{i-1};a,b) + N(a_{i+1}p ;a,b)\\
							&=~~~~ N(p a_{i-1};a,b) + N(\gamma(\alpha_{i-1}) ;a,b) + N(b_{i-1}p';a,b) \\
							& ~~~~ - N(p a_{i+1} ;a,b) - N(\gamma(\alpha_{i+1});a,b) - N( b_{i+1}p';a,b)\\
							&=~~~ N(\beta_{i-1};a,b)- N(\beta_{i+1};a,b).
	\end{align*}
	It follows that, using indices modulo $3$,
	\begin{align*}
		\sum_{i=0}^2 N(\gamma(\pi_i);a,b) 
			~~=~~ \sum_{i=0}^2\left(N(\beta_{i-1};a,b) - N(\beta_{i+1};a,b) \right)
			~~=~~ 0. 
	\end{align*}
	Therefore, if $N(\gamma(\pi_i);a,b)\not= 0$ for some $i$, at least another cycle $\pi_j$, $j\not= i$,
	must have $N(\gamma(\pi_j);a,b)\not= 0$.\qedhere
\end{proof}

When a family of closed walks satisfies the 3-path condition, there is a general method to find a shortest element in the family.
The method is based on considering fundamental-cycles defined by shortest-path trees, which is
precisely what our algorithm is doing specialized for the family $\Pi$. 
We thus obtain:

\begin{lemma}\label{le:shortest}
	The cycle $\tau^*$ is a shortest element of $\Pi$.
\end{lemma}
\begin{proof}
It is a consequence of the 3-path condition, 
	that a shortest cycle in 
	\[
		\{ \cycle(T_r,e)\mid r\in S, e\in E(\I)\setminus E(T_r), \cycle(T_r,e)\in \Pi \} 
		~~=~~ 
		\{ \cycle(T_r,e)\mid (r,e)\in P \} 
	\]
	is a shortest cycle in $\Pi$. That is, the search for a shortest element in $\Pi$
	can be restricted to cycles of the type $\cycle(T_r,e)$.
	See Thomassen~\cite{t-egsnc-90} or the book by Mohar and Thomassen~\cite[Chapter 4]{mt-gs-01} 
	for the so-called fundamental cycle method. 
	(The method is described for unweighted graphs but it also works for weighted graphs.
	See, for example, Cabello et al.~\cite{ccl-fsntc-10} for the generalized case of weighted, directed graphs.)
\end{proof}

The next step in our argument is showing that $\gamma^*$ is simple (without self-intersections) 
and separates $a$ and $b$.
We will use the following characterization of which simple, closed  
polygonal paths separate $a$ and $b$. 

\begin{lemma}\label{le:jordan}
	For any simple, closed polygonal path $\gamma$ we have $|N(\gamma;a,b)|\le 1$.
	Furthermore, $\gamma$ separates $a$ and $b$ if and only if $N(\gamma;a,b)= \pm 1$.
\end{lemma}
\begin{proof} 
	Since $\gamma$ is simple, it defines an interior and an exterior by the Jordan curve theorem.
	The crossings between $\gamma$ and $\segment{ab}$, as we walk along $\segment{ab}$,
	alternate between left-to-right and right-to-left crossings because $\segment{ab}$ has
	pieces alternating in the interior and exterior of $\gamma$.
	Therefore $|N(\gamma;a,b)|\le 1$.
	
	Assume that $\gamma$ separates $a$ and $b$, so that one is in the interior of $\gamma$
	and the other in the exterior.
	Then the segment $\segment{ab}$ crosses $\gamma$ an odd number of times, and it must be 
	$|N(\gamma;a,b)|=1$. Conversely, if $|N(\gamma;a,b)|=1$, then the number of intersections between
	$\gamma$ and $\segment{ab}$ is odd, which implies that one of the points $a$ and $b$ is in the interior
	of $\gamma$ and the other in the exterior.
\end{proof}

We can now prove that $\gamma^*$ is simple using a standard uncrossing argument.
Indeed, a self-crossing of $\gamma^*$ would imply that we can find a strictly 
shorter element in $\Pi$, which would contradict the property stated in Lemma~\ref{le:shortest}.

\begin{lemma}\label{le:feasible}
	The polygonal path $\gamma^*$ is simple and separates $a$ and $b$.
\end{lemma}
\begin{proof}
Assume, for the sake of contradiction, that $\gamma^*$ is not simple.
	It is then possible to show the existence of two cycles $\tau_1$ and $\tau_2$ in $\I$ such that
	$\length(\tau_1)<\length(\tau^*)$,
	$\length(\tau_2)<\length(\tau^*)$,
	and $N(\gamma^*;a,b)= N(\gamma(\tau_1);a,b)+ N(\gamma(\tau_2);a,b)$,
	as follows.
	
	\begin{figure}
		\centering
			\includegraphics[width=0.8\textwidth]{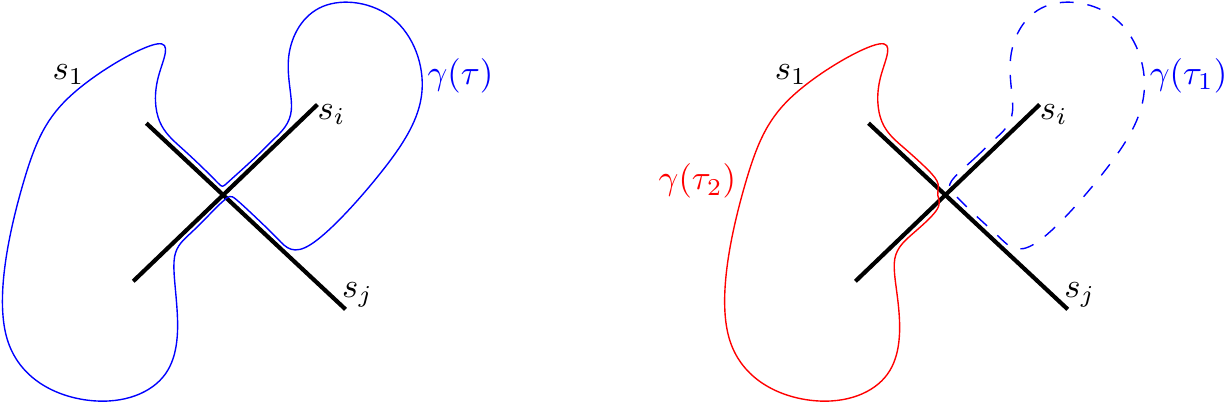}
			\caption{The polygonal paths defined by the cycles $\tau_1$ and $\tau_2$ from the cycle $\tau$ in Lemma~\ref{le:feasible}.}
			\label{fig:partition}
	\end{figure}	

	Let $s_0 s_1s_2\dots s_t$, with $s_t=s_0$, be the cycle $\tau^*$.
	Start walking along $\gamma^*$ from $s_0\cap s_1$, until we find the first self-intersection,
	which is defined by segments $s_i$ and $s_j$, with $i< j$.
	Note that $2\le j-i$ because $s_i$ and $s_{i+1}$ cannot define a self-intersection of $\gamma^*$.
	Consider the cycles $\tau_1 = s_i s_{i+1}\dots s_j s_i$
	and $\tau_2=s_0\dots s_i s_j \dots s_t$. See Fig.~\ref{fig:partition}.
	Note that
	\[
		N(\gamma^*;a,b) ~~=~~ N(\gamma(\tau_1);a,b) + N(\gamma(\tau_2);a,b) 
	\]
	because the polygonal paths $\gamma(\tau_1)$ and $\gamma(\tau_2)$ form a disjoint partition
	of $\gamma^*$, with orientations preserved.
	Moreover, because $j-i\ge 2$ and $\tau^*$ is a cycle,
	we have $\length(\tau_1)< \length (\tau^*)$ and $\length(\tau_2)< \length (\tau^*)$.
	This finishes the proof of existence of $\tau_1$ and $\tau_2$.

	Because $\tau^*\in \Pi$ we have 
	\[
		0~~ \not= ~~ N(\gamma^*;a,b)~~ =~~ N(\gamma(\tau_1);a,b) + N(\gamma(\tau_2);a,b) .
	\]
	Therefore, $N(\gamma(\tau');a,b)\not = 0$ for some $\tau'\in\{ \tau_1,\tau_2\}$. 
	Since $\length(\tau')<\length (\tau^*)$ and $N(\gamma(\tau');a,b)\\ \not = 0$,
	then $\tau'\in \Pi$. This contradicts the property that $\tau^*$ is a shortest
	cycle of $\Pi$ (Lemma~\ref{le:shortest}).
	We conclude that $\gamma^*$ must be simple.
		
	Since $\gamma^*$ is simple, $N(\gamma^*)\in \{ -1,0,+1\}$ by Lemma~\ref{le:jordan}.
	Since $\tau^*\in \Pi$, then $N(\gamma^*)\not = 0$, which implies $N(\gamma^*)= \pm 1$.
	It then follows from Lemma~\ref{le:jordan} that $\gamma^*$ separates $a$ and $b$.\qedhere	
\end{proof}

We can now prove the main theorem.
\begin{theorem}\label{thm:separation}
	The weighted version of {\sc 2-Cells-Separation} can be solved in $\mathcal{O}(nk+n^2\log n)$ time,
	where $n$ is the number of input segments and $k$ is the number of pairs of segments that intersect.
\end{theorem}
\begin{proof}
	We use the algorithm described in Section~\ref{sec:algorithm}.
	The algorithm returns a feasible solution because of Lemma~\ref{le:feasible}:
	the cycle $\gamma^*$ separates $a$ and $b$ and
	is contained in $\bigcup S(\cycle(T_{r^*},e^*))$, therefore, 
	the set $S(\cycle(T_{r^*},e^*)$ returned by the algorithm separates $a$ and $b$.
	
	To see the optimality of the weight of $S(\tau^*)$, consider an optimal solution $S_*\subseteq S$.
	Assume that we run the algorithm on $S_*$.
	The algorithm would compute a cycle $\tau_*$ in the intersection graph of the segments $S_*$
	and return $S(\tau_*)\subseteq S_*$. By Lemma~\ref{le:feasible},
	the polygonal path $\gamma(\tau_*)$ is simple and separates $a$ and $b$.
	Lemma~\ref{le:jordan} implies that $N(\gamma(\tau_*);a,b)=\pm 1 \not =0$,
	and therefore $\tau_*\in \Pi$ (here $\Pi$ refers to the original problem, rather than the subproblem
	defined by input $S_*$).
	
	For any cycle $\pi$ of $\I$ we have $\length(\pi)= 2 |S(\pi)|$ because of the choice of the edge-weights in $\I$. 
	Since $\tau^*$ is a shortest cycle in $\Pi$ by Lemma~\ref{le:shortest},
	we have
	\[ 
		w(S(\tau^*))~~=~~\tfrac 12 \length (\tau^*) 
					~~\le~~ \tfrac 12 \length(\tau_*) 
					~~=~~ w(S(\tau_*))
					~~\le~~ w(S_*).
	\]
	It follows that $S(\tau^*)$ is a feasible solution whose weight is not larger than $w(S_*)$,
	and therefore it is optimal.
	The running time follows from Lemma~\ref{le:time}.
\end{proof}

\begin{corollary}\label{co:separation}
	The weighted version of {\sc $2$-Cells-Separation} in which the segments have weights $0$ or $1$ 
	can be solved in $\mathcal{O}(n^2+nk)$ time,
	where $n$ is the number of input segments and $k$ is the number of pairs of segments that intersect.
\end{corollary}
\begin{proof}
	In the proof of the previous theorem we use Lemma~\ref{le:time2} instead of Lemma~\ref{le:time}.
\end{proof}

In the case where the segments of $S$ are unweighted, the points $a, b$ are inside a polygon $P$ with holes, and the $a$-$b$ path must be contained in the interior of $P$, the problem can be easily solved by 
assigning weight 0 to the edges $E(P)$ of the polygon $P$
and weight $1$ to the segments in $S$. 
We can then apply Corollary~\ref{co:separation} on $S\cup E(P)$, and obtain the following.

\begin{corollary}
The restricted {\sc $2$-Cells-Separation} problem in a polygon with holes can be solved in $\mathcal{O}(n^2+nk)$ time, where $n$ is the total size of the input
and $k$ is the number of pairs of segments in $S$ that intersect.
\end{corollary}


\section{Connecting two cells}
\label{2cells}

We show that {\sc $2$-Cells-Connection} is NP-hard and APX-hard by a reduction from {\sc Exact-Max-$2$-Sat}, a well studied NP-complete and APX-complete problem(c.f.~\cite{Hastad01}):
Given a propositional CNF formula $\Phi$ with $m$ clauses on $n$ variables and exactly two variables per clause, decide whether there exists a truth assignment that satisfies at least $k$ clauses, for a given $k\in\mathbb{N}$, $k\leq m$.
Let $x_1,\dots,x_n$ be the variables of $\Phi$, $\ell_i$ be the number of appearances of variable $x_i$ in $\Phi$, and $\ell= \sum_i \ell_i$; since each clause contains exactly 2 variables, $\ell=2m$. The maximum number of satisfiable clauses is denoted by ${\rm opt}(\Phi)$.
Using $\Phi$ we construct an instance consisting of a set of segments $S=S(\Phi)$ and two points $a=a(\Phi)$ and $b=b(\Phi)$ as follows.

Abusing the terminology slightly, the term \emph{segment} will refer to a set of identical single segments stacked on top of each other. The cardinality of the set is the \emph{weight} of the segment. Either all or none of the single segments in the set can be crossed by a path. 
There are two different types of segments, $\tau_1$, and $\tau_{\infty}$, according to their weight. Segments of type $\tau_1$ have weight $1$ (light or single segments), while segments of type $\tau_{\infty}$ have weight $20m$ (heavy segments). The weight of heavy segments is chosen so that they are never crossed by an optimal $a$-$b$ path.

\begin{figure}
\centering
\includegraphics[width=.43\textwidth]{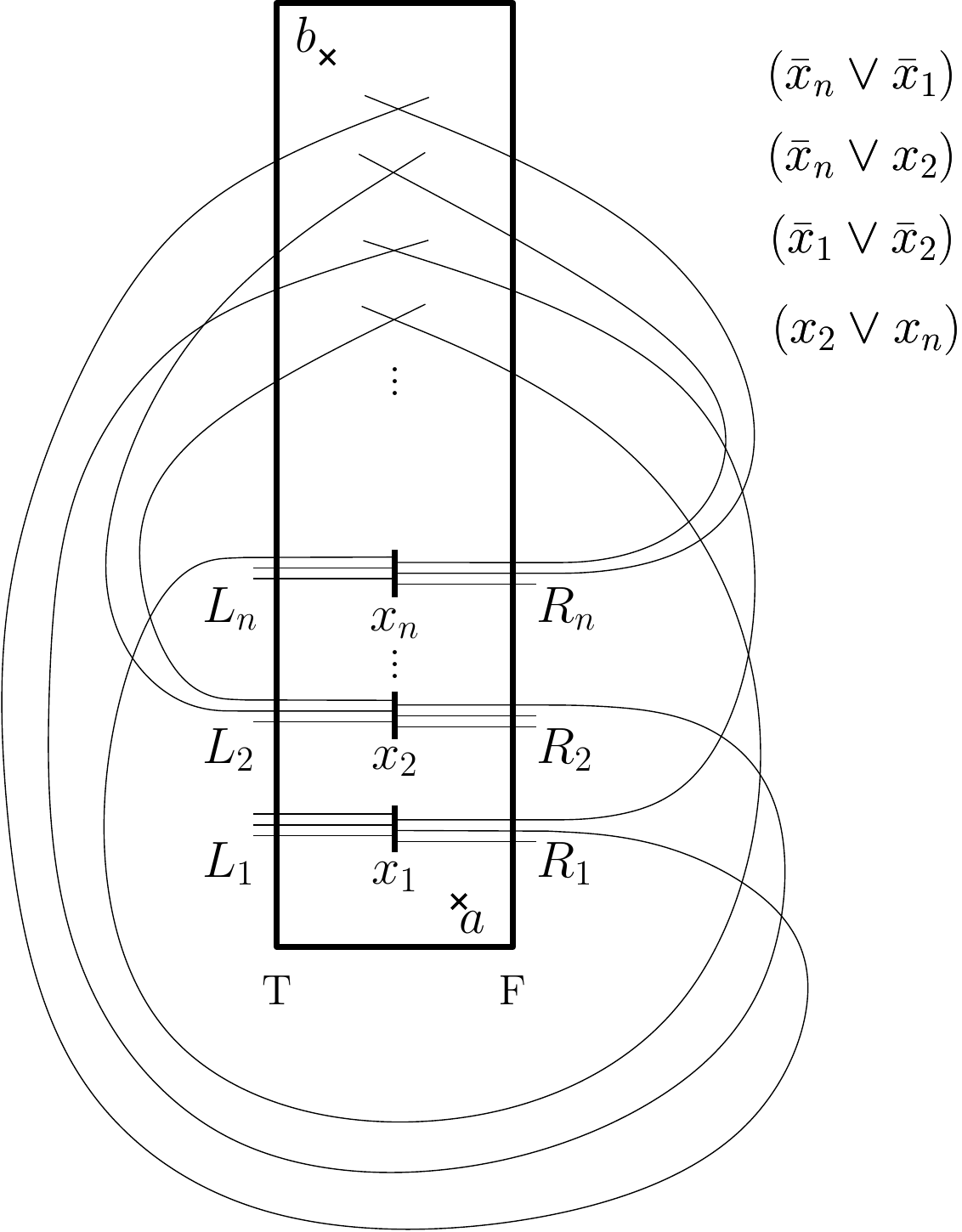}
\caption{Idea of the construction with curved segments.}
\label{idea}
\end{figure}

We first provide an informal, high-level description of the construction that uses \emph{curved} segments. Later on, each curved segment will by replaced by a collection of straight-line segments in an appropriate manner. See Fig.~\ref{idea}. We have a rectangle $R_\infty$ made of heavy segments, with point $a$ at a lower corner and $b$ at an upper corner. For each variable $x_i$ we add a small vertical segment of type $\tau_{\infty}$ in the lower half of $R_\infty$.
From the segment we place $\ell_i$ horizontal light segments, denoted by $R_i$, going to the right and $\ell_i$ horizontal light segments, denoted by $L_i$, going to the left until they reach the outside of $R_{\infty}$. Roughly speaking, (things are slightly more complicated)
an optimal $a$-$b$ path will have to choose for each $x_i$ whether it crosses all segments in $L_i$, encoding the assignment $x_i=\true$, or all segments in $R_i$, encoding the assignment $x_i=\false$.
Consider a clause like $x_2\vee x_{n}$, where both literals are positive.  We prolong one of the segments of $L_2$ and one of the segments of $L_n$ with a curved segment so that they cross again inside $R_\infty$ (upper half) in such a way that an $a$-$b$ path inside $R_\infty$ must cross one of the prolongations, and one is enough; see Fig.~\ref{idea}, where one of the prolongations passes below $R_\infty$. A clause like $\bar x_n\vee x_2$ is represented using prolongations of one segment from $L_2$ and one segment of $R_n$. The other types of clauses are symmetric. For each clause we always prolong different segments; since $L_i$ and $R_i$ have $\ell_i$ segments, there is always some segment that can be prolonged. It will then be possible to argue that the optimal $a$-$b$ path has cost $\ell+ (m-{\rm opt}(\Phi))$. 
We do not provide a careful argument of this here since we will need it later for a most complicated scenario. This finishes the informal description of the idea.

We now describe in detail the construction with straight-line segments.
First, we construct a polygon, called the \emph{tunnel}, with heavy boundary segments of type $\tau_{\infty}$; see Fig.~\ref{chain}(a).  
The tunnel has a `zig-zag' shape and can be seen as having $8$ corridors, $C_1,\ldots ,C_8$. It starts with $C_1$,  the \emph{main} corridor (at the center of the figure), which contains point $a$, then it turns left to  $C_2$, then right, etc., gradually turning around to $C_7$ and then to the \emph{end} corridor $C_8$ (at the top). The latter contains point $b$. To facilitate the discussion, we place a point $b'$ in the tunnel where the transition from $C_7$ to the end corridor occurs. 
The tunnel has a total weight of $21\cdot 20m =\mathcal{O}(m)$.
The rest of the construction will force any $a$-$b$ path of some particular cost (to be given shortly) to stay always in the interior of the tunnel. 

\begin{figure}
\centering
\includegraphics[width=.88\textwidth]{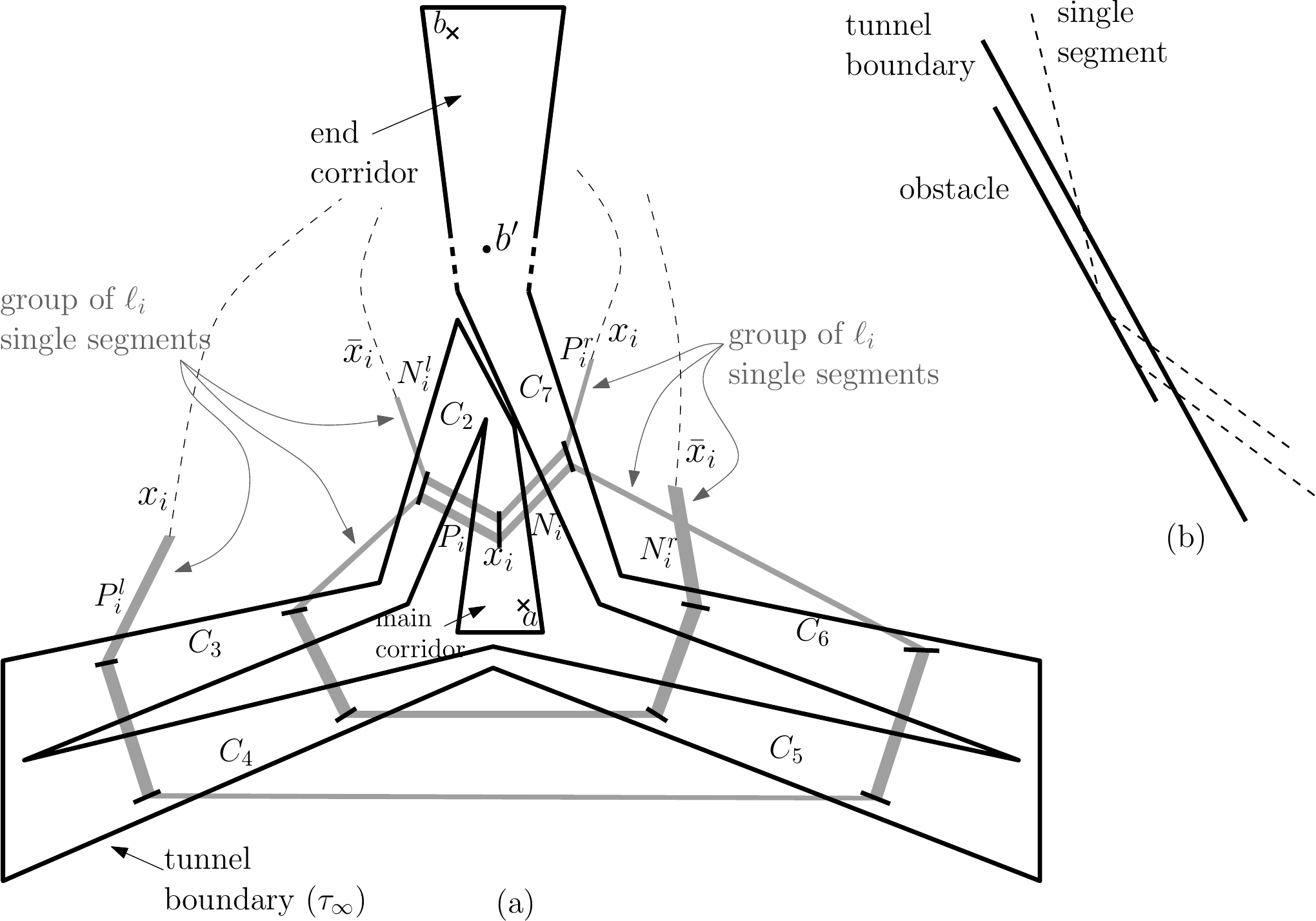}
\caption{(a) Tunnel and variable chain. Each gray trapezoid represents a piece with $\ell_i$ parallel segments. (b) Part of a chain piece close to the tunnel.}
\label{chain}
\end{figure}

Each variable $x_i$ of $\Phi$ is represented by a collection of 16 pieces, which form a chain-like structure. Each \emph{piece} is a group of $\ell_i$ nearly-parallel single segments whose ends are either outside the tunnel or lie on `short' heavy segments of type $\tau_{\infty}$ in the interior of the tunnel, referred to as \emph{obstacles}. For each variable, there is one obstacle in each of the corridors $C_1$, $C_2$, $C_7$ and there are two obstacles in each of the corridors $C_3$, $C_4$, $C_5$, and $C_6$. See Fig.~\ref{chain}(a), where we represent each piece by a light gray trapezoid and each obstacle by a bold, short segment. Pieces always contain a part outside the tunnel. The exact description of the structure is cumbersome; we refer the reader to the figures. The obstacle in $C_2$ contains the extremes of four pieces: two pieces, called $P_i$, go to the obstacle in the main corridor, one goes to an obstacle in $C_3$, and the fourth piece, which we call $N_i^l$ goes outside the tunnel. Symmetrically, the obstacle in $C_7$ contains the extremes of four pieces: two pieces, called $N_i$, go to the main corridor, one goes to the corridor $C_6$, and one, which we call $P_i^r$ goes outside the tunnel. We add pieces connecting the obstacles in $C_3$ and $C_4$, the obstacles in  $C_4$ and $C_5$, and the obstacles in $C_5$ and $C_6$. From the obstacle in $C_3$ that currently has one piece we add another piece, which we call $P_i^l$ and whose other extreme is outside the tunnel. From the obstacle in $C_6$ that currently has one piece we add another piece, which we call $N_i^r$, whose other extreme is outside the tunnel. 

The obstacles and the pieces of all variables should satisfy some conditions: obstacles should be disjoint, pieces can touch only the obstacles at their extremes, and pieces may cross only outside the tunnel. See Fig.~\ref{complete}. Some of the single segments of $P_i^r$, $P_i^l$, $N_i^r$, $N_i^l$ will be prolonged and rotated slightly to encode the clauses. For this, we will need that the line supporting a segment from $P_i^r\cup N_i^r$ intersects inside the end corridor the line supporting a segment from $P_j^l\cup N_j^l$. This can be achieved by stretching the end corridor sufficiently and placing the obstacles of $C_2$ and $C_7$ close to the tunnel boundary; see Fig.~\ref{chain}(b).

\begin{figure}
\centering
\includegraphics[width=.86\textwidth]{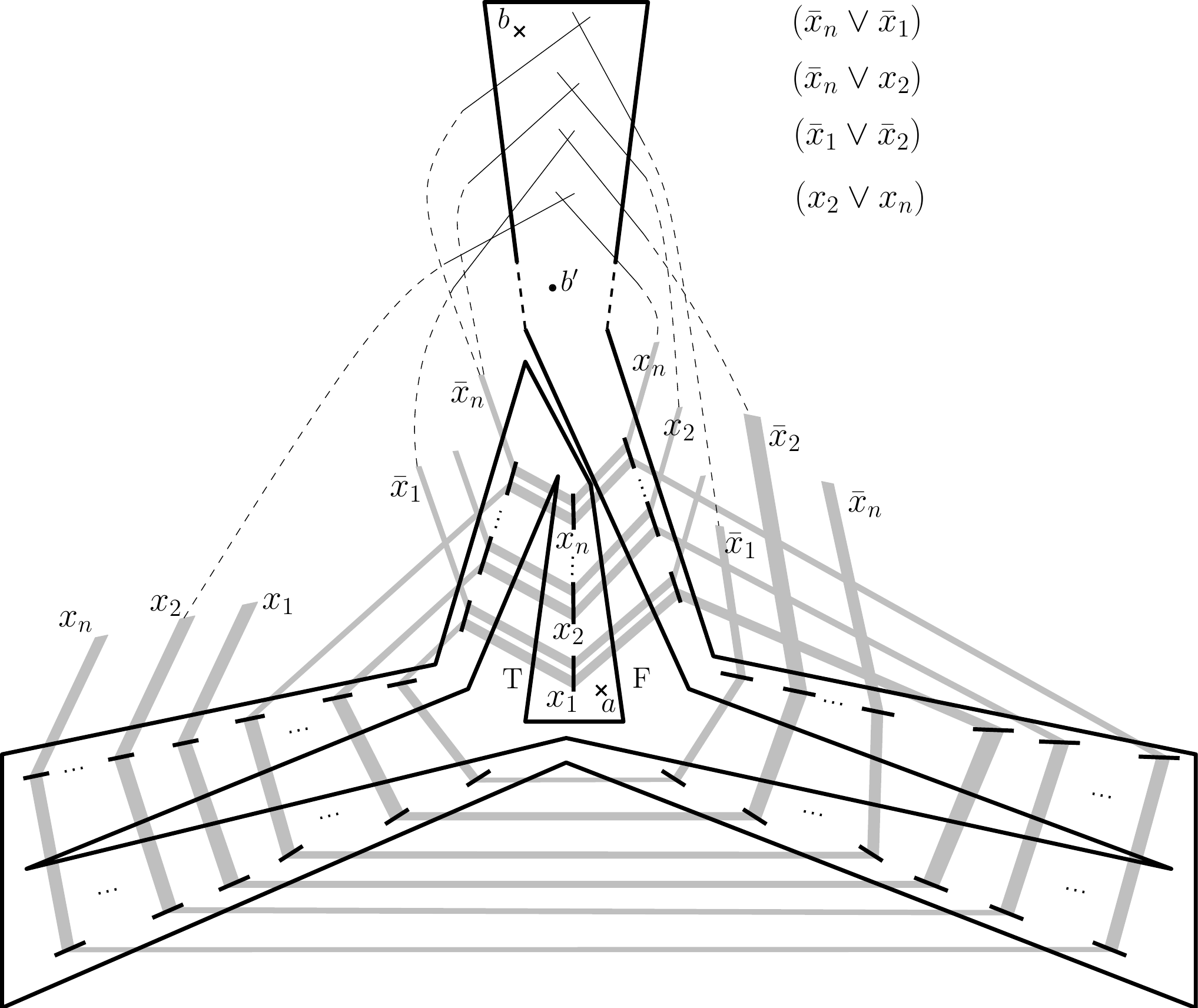}
\caption{Example of overall construction.}
\label{complete}
\end{figure}

For each clause of $\Phi$ we prolong two segments of $P_i^r\cup P_i^l \cup  N_i^r\cup N_i^l$ as follows; see Fig.~\ref{complete} for an example of the overall construction, where prolongations are shown by dashed lines. Each segment corresponds to some literal $x_i$ or $\bar x_i$ in the clause: in the first case the segment comes from either $P_i^r$ or $P_i^l$, while in the second one it comes from either $N_i^r$ or $N_i^l$. For the construction, these choices for each clause can be made arbitrarily, provided that one segment intersects the tunnel from the left side and the other one from the right. These segments are prolonged until their intersection point inside the end corridor. For each clause, two different segments are prolonged. Since the pieces corresponding to variable $x_i$ have $\ell_i$ segments, there is always some segment available. 
Segments corresponding to different clauses may intersect only outside the tunnel; this is ensured by rotating the segments slightly around the endpoint lying in the obstacle. In this way, the end corridor is obstructed by $m$ pairs of intersecting segments such that any path from the intermediate point $b'$ to point $b$ staying inside the tunnel must intersect at least one segment from each pair.  

The following lemma establishes the correctness of the reduction. 

\begin{lemma}\label{sec2_lem:bothdirections}
There is an $a$-$b$ path of cost at most $8\ell + k$, where $1\le k \le m$, if and only if 
there is a truth assignment satisfying at least $(m-k)$ of the clauses.
\end{lemma}
\begin{proof}
We denote by $S_i$ the set of segments in the pieces corresponding to the variable $x_i$.
We denote by $S_i^{\rm T}$ the segments in the pieces of $P_i$, the piece connecting $C_7$ to $C_6$, the piece  $P_i^r$, and so on in an alternating manner along the chain structure. Note that $S_i^{\rm T}$ contains $P_i$, $P_i^l$ and $P_i^r$.
We denote by $S_i^{\rm F}$ the segments $S_i\setminus S_i^{\rm T}$. Note that $S_i^{\rm F}$ contains $N_i$, $N_i^l$ and $N_i^r$.
See Fig.~\ref{chain-posneg}.
Each of the sets $S_i^{\rm T}$ and $S_i^{\rm F}$ contains $8\ell_i$ segments. Inside the tunnel there is an $a$-$b'$ path disjoint from $S_i^{\rm T}$ and there is another $a$-$b'$ path disjoint from $S_i^{\rm F}$. We also denote by $T_j$ the two segments used for clause $C_j$ of $\Phi$.

\begin{figure}
\centering
\hfill
\includegraphics[width=.49\textwidth]{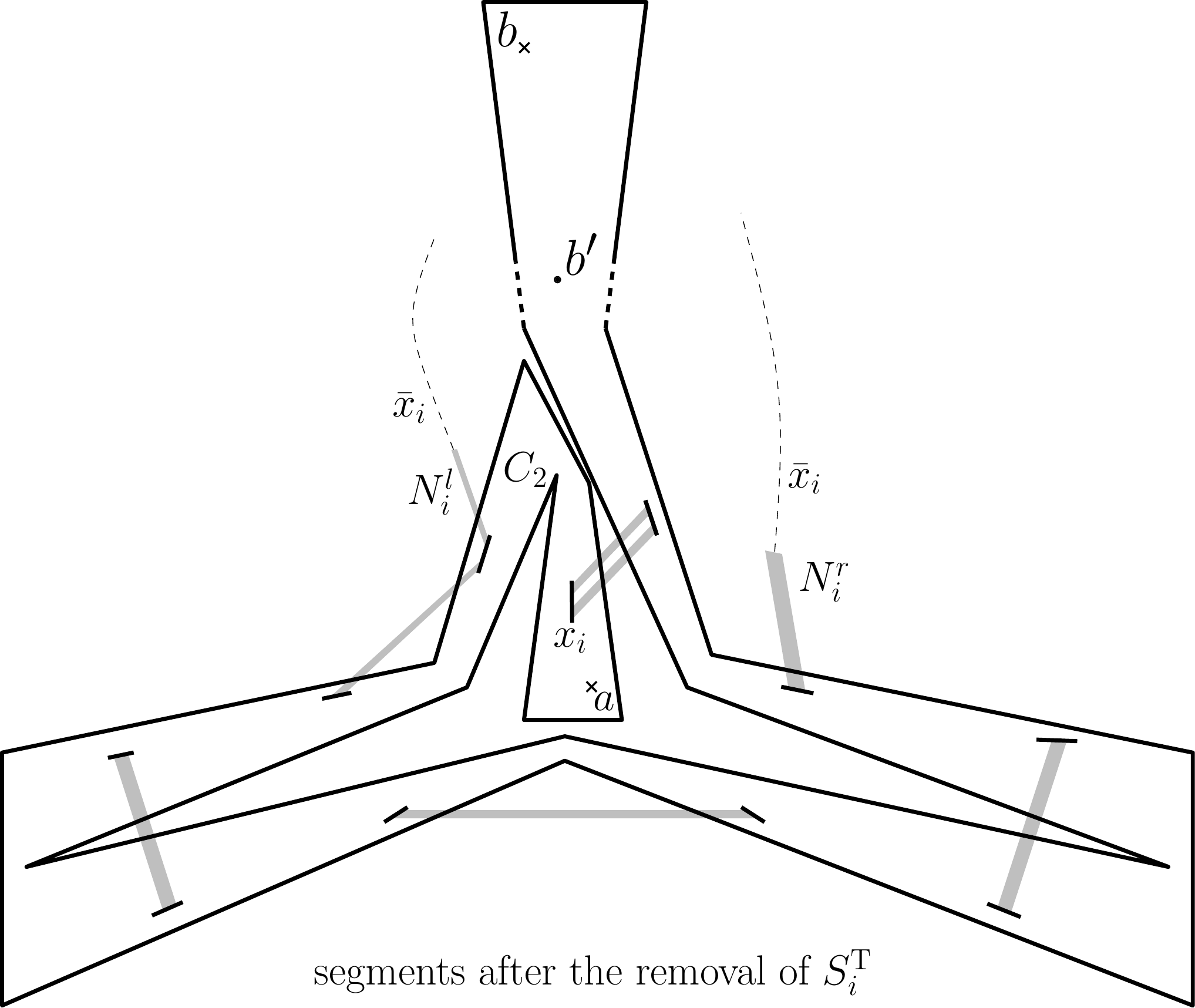}
\centering
\hfill
\includegraphics[width=.49\textwidth]{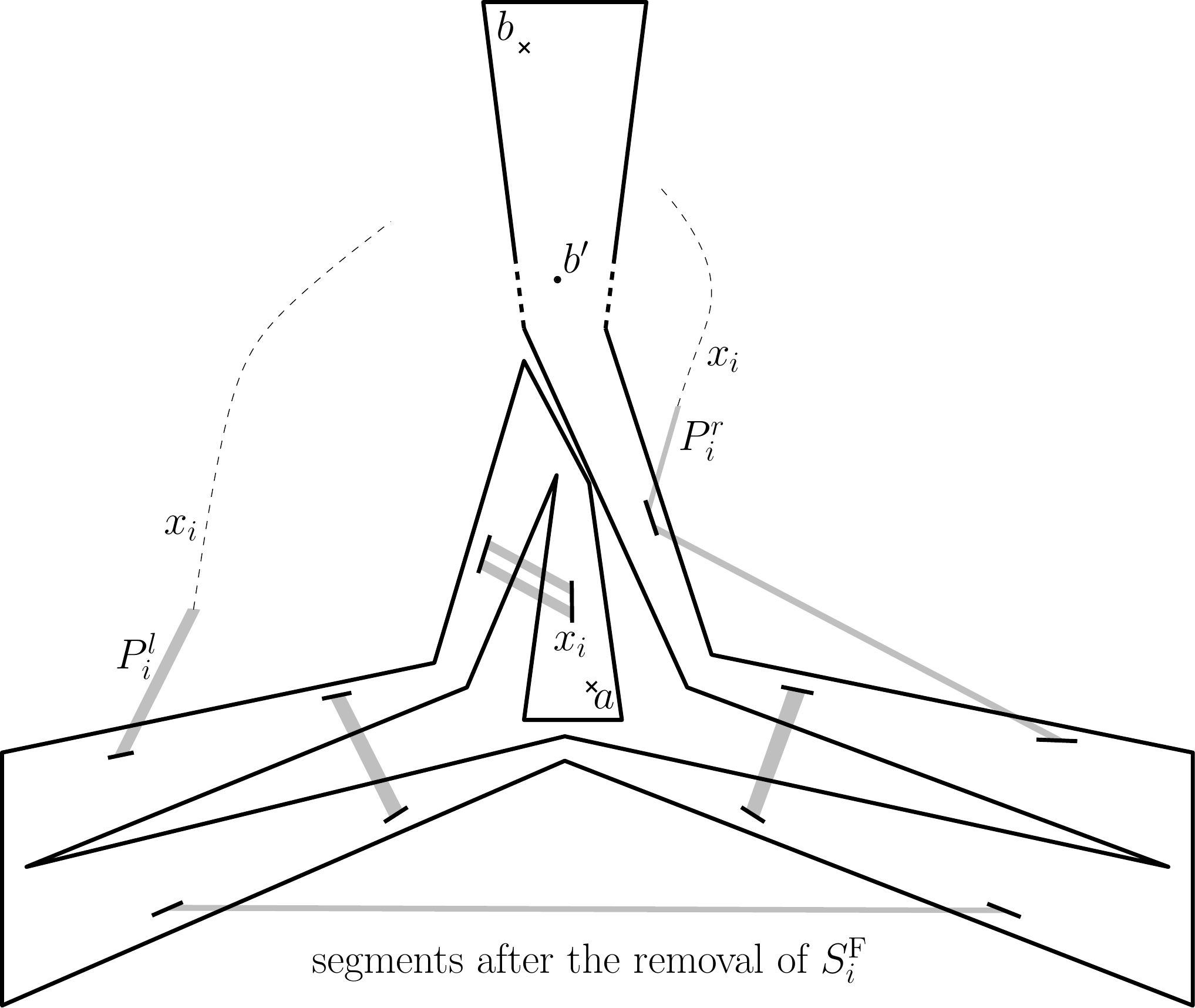}
\hfill
\caption{Removal of $S_i^{\rm T}$ (left) and $S_i^{\rm F}$ (right).}
\label{chain-posneg}
\end{figure}

Consider a truth assignment $\{x_i=b_i\}$, where each $b_i\in \{\true ,\false \}$, satisfying at least $(m-k)$ clauses. We construct a subset of segments $S'$ where we include the set $S_i^{b_i}$, for each variable $x_i$, and a segment of $T_j$, for each clause $C_j$ that is not satisfied by the truth assignment. Since $|S_i^{b_i}|=8\ell_i$, the set $S'$ contains at most $8\ell+ k$ segments. The removal of $S'$ leaves the points $a$ and $b'$ in the same cell of the arrangement. Equivalently, there is an $a$-$b'$ path inside the tunnel that crosses only segments from $S'$. If a clause $C_j$ of $\Phi$ is satisfied by the truth assignment, then at least one of the segments in $T_j$ is included in $S_i^{b_i}\subset S'$. If a clause $C_j$ is not satisfied, then one of the segments $T_j$ is included in $S'$ by construction. Thus, for each clause $C_j$ we have $T_j\cap S'\not= \emptyset$. It follows that $b'$ and $b$ are in the same cell after the removal of $S'$.

Conversely, note first that any $a$-$b$ path with cost at most $8\ell +k\le  16 m+m = 17m$ cannot intersect the tunnel boundary or an obstacle because segments of type $\tau_{\infty}$ have weight $20m$. Let $S'$ be the set of segments crossed by the path. If $P_i\subset S'$, then we define $b_i=\true$; otherwise, we define $b_i=\false$.
Note that when $P_i\not\subset S'$, then $N_i\subset S'$ because the $a$-$b$ path is inside the tunnel. (However it may be $N_i\cup P_i\subset S'$, so the assignment of $b_i$ is not symmetric.) We next argue that the truth assignment $\{x_i=b_i\}$ satisfies at least $(m-k)$ clauses.

Consider the case when $P_i\subset S'$. Inspection shows that 
\[
	|S'\cap S_i| \ge 8\ell_i + |S'\cap (N_i^l \cup  N_i^r)|.
\]
Indeed, after the removal of $P_i\cup N_i^l \cup  N_i^r$ any path from $a$ to $b'$ must still cross at least $6$ pieces. Similarly, inspection shows that when $N_i\subset S'$ we have 
\[
	|S'\cap S_i| \ge 8\ell_i + |S'\cap (P_i^l \cup  P_i^r)|.
\] 
Let $A_i= N_i^l \cup  N_i^r$ if $b_i=\true$ and $A_i= P_i^l \cup  P_i^r$ if $b_i=\false$. The previous cases can be summarized as 
\[
	|S'\cap S_i| \ge 8\ell_i + |S'\cap A_i|.
\]
We further define
\[
	Y= \bigcup_i (S'\cap A_i).
\]
For each clause $C_j$ we have $S'\cap T_j\not= \emptyset$ by construction, as otherwise $a$ and $b$ cannot be in the same cell of $S\setminus S'$. If $C_j$ is not satisfied by the truth assignment $\{x_i=b_i\}$, then it must be $(S'\cap T_j)\subset S'\cap A_k$ for some variable $x_k$ in $C_j$. This means that  
$T_j\cap Y\not=\emptyset$. Since the sets $T_j$ are disjoint by construction, the number of unsatisfied clauses is bounded by $|Y|$.
Using that 
\[
	8\ell + k~ = ~ |S'| ~=~ \sum_{i=1}^n |S' \cap S_i| ~\ge~ \sum_{i=1}^n (8\ell_i + |S'\cap A_i|) ~=~ 8\ell + \sum_{i=1}^n |S'\cap A_i|,
\]
we obtain 
\[
	\sum_{i=1}^n |S'\cap A_i| \le k.
\]
Therefore, the total number of clauses with value $\false$ is bounded by
\[
	|Y| ~=~ \sum_i |S'\cap A_i| ~\le~ k.\qedhere
\]
\end{proof}
The construction can be easily modified by replacing every heavy segment with a set of $20m$ distinct parallel single segments such that every single segment in $S$ that originally intersected the heavy segment now intersects all the segments in the new set and such that no three segments have a point in common. We have the following:

\begin{theorem}\label{NP-hard}
 {\sc $2$-Cells-Connection} is NP-hard and APX-hard even when no three segments intersect at a point.
\end{theorem}
\begin{proof}
NP-hardness follows form Lemma~\ref{sec2_lem:bothdirections} and the fact that the reduction produces $\mathcal{O}(nm)$ segments, whose coordinates can be bounded by a polynomial in $(n+m)$. APX-hardness follows from the fact that the reduction is approximation-preserving, as we now show.

	First, since there is always an assignment that satisfies at least $3m/4$ clauses, we have that $m\le (4/3) {\rm opt}(\Phi)$. Recall that an optimal $a$-$b$ path costs $8\ell+ (m-{\rm opt}(\Phi))$, where $\ell=2m$. A polynomial-time $c$-approximation algorithm ($c>1$) for the problem would give a path that costs at most

	\begin{align*}
	c(8\ell+ (m-{\rm opt}(\Phi))) ~~&=~~  c(17m+{\rm opt}(\Phi))\\
									 &=~~  17m-c\,{\rm opt}(\Phi)+17(c-1)m\\
													&\le~~  17m-c\,{\rm opt}(\Phi) + 17(c-1)(4/3){\rm opt}(\Phi)\\
													&=~~  17m-{\rm opt}(\Phi)(68/3-(65/3)c)\\
													&=~~   16m+\Bigl[ m-{\rm opt}(\Phi)(68/3-65c/3) \Bigr]
	\end{align*}
	and, by Lemma~\ref{sec2_lem:bothdirections}, a truth assignment that satisfies at least 
	${\rm opt}(\Phi)(68/3-65c/3)$ clauses. However, {\sc Exact-Max-2-SAT} cannot be approximated 
	above $21/22$~\cite{Hastad01}, which implies that $c$ must be larger than $(68/65 - 63/(22\cdot 65))\approx 1.002097\dots$
	(A slightly better inapproximability result can be obtained using the better bounds that rely on the unique games conjecture~\cite{kkmo-oir-07}.)
\end{proof}

We can reduce {\sc $2$-Cells-Connection} to the minimum color path problem (MCP): 
Given a graph G with colored (or labeled) edges and two of its vertices, find a path between the vertices that uses the minimum possible number of colors.
We color the edges of the dual graph $G$ of $\mathcal{A}(S)$ as follows: two edges of $G$ get the same color if and only if their corresponding edges in $\mathcal{A}(S)$ lie on the same segment of $S$.
Then, finding an $a$-$b$ path of cost $k$ in $\mathcal{A}(S)$ amounts to finding a $k$-color path in $G$ between the two cells which $a$, $b$ lie in.

However, MCP is NP-hard~\cite{BLWZ05} and W[1]-hard~\cite{FGK10} (with respect to solution size) even for planar graphs, it has a polynomial-time $\mathcal{O}(\sqrt{n})$-approximation algorithm and is non-approximable within any polylogarithmic factor~\cite{HMS07}. 

\section{Tractable cases for connecting two cells}

We now describe two special cases where {\sc $2$-Cells-Connection} is tractable. First, we consider the case where the input segments
have few crossings, in a sense that is specified below. Then, we return to the special case where we have a polygon and provide an algorithm
that takes polynomial time when the number of holes in the polygon is constant.
 
\subsection{Segments crossings.}
Without loss of generality, we assume that every segment in $S$ intersects at least two other segments and that both endpoints of a segment are intersection points. We say that two segments cross if and only if they intersect at a point that is interior to both segments (a segment crossing).

Consider the colored dual graph $G$ of $\mathcal{A}(S)$ as defined after Theorem~\ref{NP-hard}. A face of $G$ (except the outer one) corresponds to a point of intersection of some $r\geq 2$ segments and has $r$ colors and, depending on the type of intersection, from $r$ to $2r$ edges. 
For example, for $r=2$ we can get two multiple edges, a triangle, or a quadrilateral, with two distinct colors. See Fig.~\ref{fig:colored_dual}(a)-(c), where the colors are given as labels.

\begin{figure}[t]
\centering
\includegraphics[width=\textwidth]{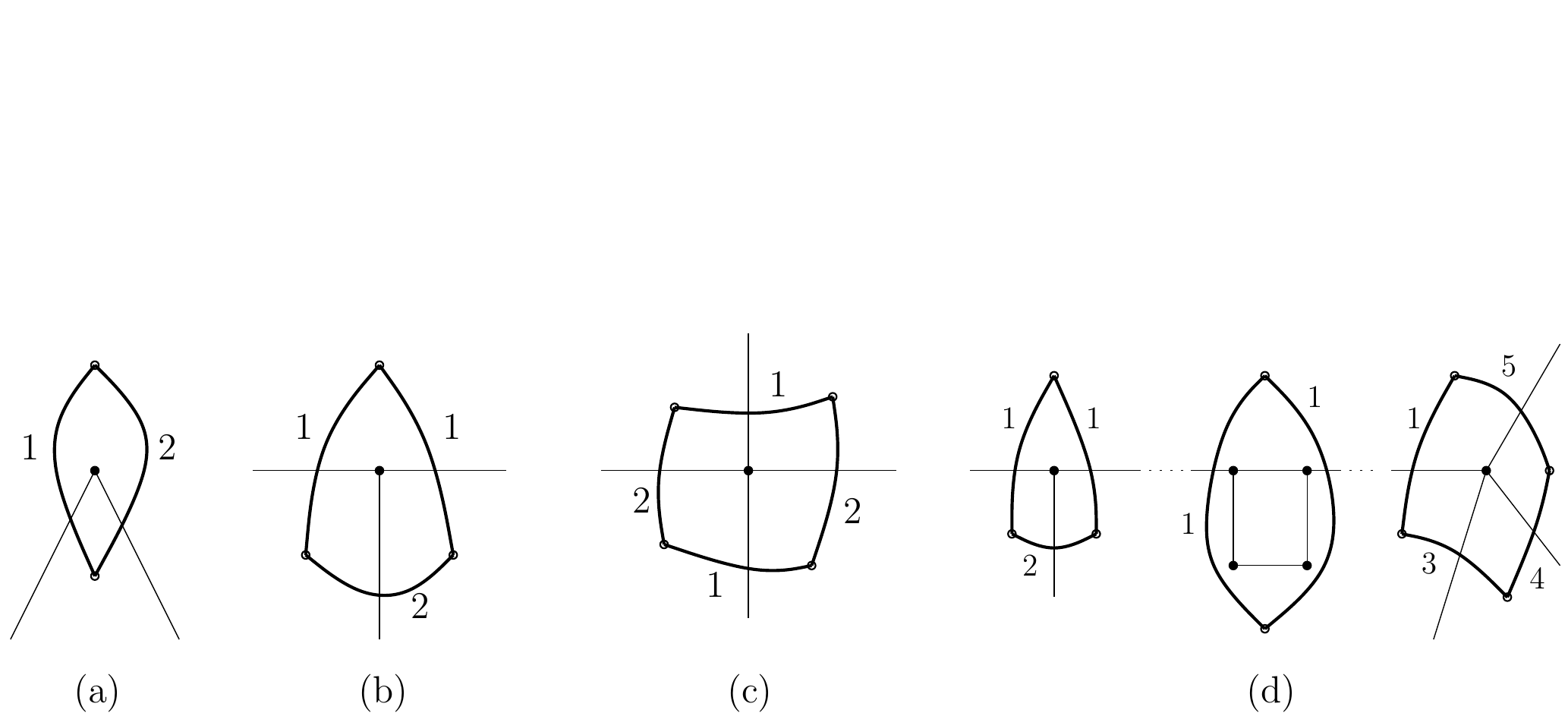}
\caption{Examples of intersections in $\mathcal{A}(S)$ and colored edges in $G$.}
\label{fig:colored_dual}
\end{figure}

When any three segments may intersect only at a common endpoint and no two segments cross, $G$ can only have multiple edges (possible all with the same color), bi-chromatic triangles, and arbitrary large faces where all edges have different colors; See Fig.~\ref{fig:colored_dual}(d) for an example. In this case, since two segments can intersect only at one point, each color induces a connected subgraph of $G$, in fact a tree (where all but one multiple edges with the same color can be deleted) for there can be no monochromatic cycle in $G$.  Then, {\sc $2$-Cells-Connection} reduces to a simple shortest path computation between the cells containing $a$ and $b$ in the (uncolored) graph resulting from $G$ by completing each monochromatic tree into a clique.
By contrast, note that {All-Cells-Connection} is still NP-hard for this special case; see Section~\ref{allcells}.

Generalizing this, if we allow $k$ segment crossings, 
we can easily reduce the problem to $2^{\mathcal{O}(k)}$ shortest path problems as follows. Let $C \subseteq S$ be the set of the (at most $2k$) segments participating in these crossings. For a fixed subset $C'$ of $C$, we first contract every edge of $G$ corresponding to a segment in $C'$, effectively putting all segments of $C'$ into the solution. Then, we delete every edge corresponding to a segment in $C\setminus C'$  that still participates in a crossing, i.e., we exclude all crossing segments of $C\setminus C'$ from the solution. 
In the resulting (possibly disconnected) graph $G'$, each of the remaining colors induces again a monochromatic subtree, thus we can compute a shortest path as before and add $C'$ to the solution. Finally, we return a minimum size solution set over all $2^{\mathcal{O}(k)}$ possible subsets $C'$. 
Thus, we have just proved the following:

\begin{theorem}
{\sc $2$-Cells-Connection} is fixed-parameter tractable with respect to the number of segment crossings if any three segments may intersect only at a common endpoint.
\end{theorem}

\subsection{Polygon with holes.}
Let $P$ be a polygon with $h$ holes and $S$ be a set of segments lying inside $P$ with their endpoints on its boundary; see Fig.~\ref{fig:polygon}. We use $n$ as a bound for the number of vertices of $P$ and segments in $S$. 
We consider the restricted {\sc $2$-Cells-Connection} problem where the $a$-$b$ path may not cross the boundary of $P$. This version is also NP-hard by a simple reduction from the general one: place a large polygon enclosing all the segments and add a hole at the endpoint of each segment. 
We assume for simplicity that $a$ and $b$ are in the interior of $P$. 

A boundary component of $P$ may be the exterior boundary or the boundary of a hole. For each boundary component $\beta$ of $P$, let $C_\beta$ be the connected component of $\mathbb{R}^2\setminus P$ that has $\beta$ as boundary, and let $z_\beta$ be an arbitrary, fixed point in the interior of $C_\beta$. If $\beta$ is the exterior boundary, then $C_\beta$ is unbounded.

Let $\beta$ and $\beta'$ be two boundary components of $P$; it may be that $\beta=\beta'$. 
Let $S_{\beta,\beta'}$ be the subset of segments from $S$ with one endpoint in $\beta$ and another endpoint in $\beta'$. We partition $S_{\beta,\beta'}$ into clusters, as follows. Consider the set $X_{\beta,\beta'}$ obtained from $P\setminus \{ a, b \}$ by adding $C_\beta$ and $C_{\beta'}$. Note that $a$ and $b$ are holes in $X_{\beta,\beta'}$. For each segment $s=\segment{pq}\in S_{\beta,\beta'}$, with $p\in \beta$ and $q\in\beta'$, we define the following
curve $\gamma_s$: follow a shortest path in $C_\beta$ from $z_\beta$ to $p$, then follow $\segment{pq}$, and then follow a shortest path
in $C_{\beta'}$ from $q$ to $z_{\beta'}$. See Fig.~\ref{fig:polygon2}.
We say that segments $s$ and $s'$ from $S_{\beta,\beta'}$ are \emph{$a$-$b$ equivalent} if $\gamma_s$ and $\gamma_{s'}$ are homotopic paths in $X_{\beta,\beta'}$.
Since being homotopic is an equivalence relation (reflexive, symmetric, transitive), being $a$-$b$ equivalent is also an equivalence relation in $S_{\beta,\beta'}$.
Therefore, we can make equivalence classes, which we call \emph{clusters}. 
The following two results provide key properties of the clusters. 

\begin{figure}[t]
\centering
\includegraphics[width=0.9\textwidth]{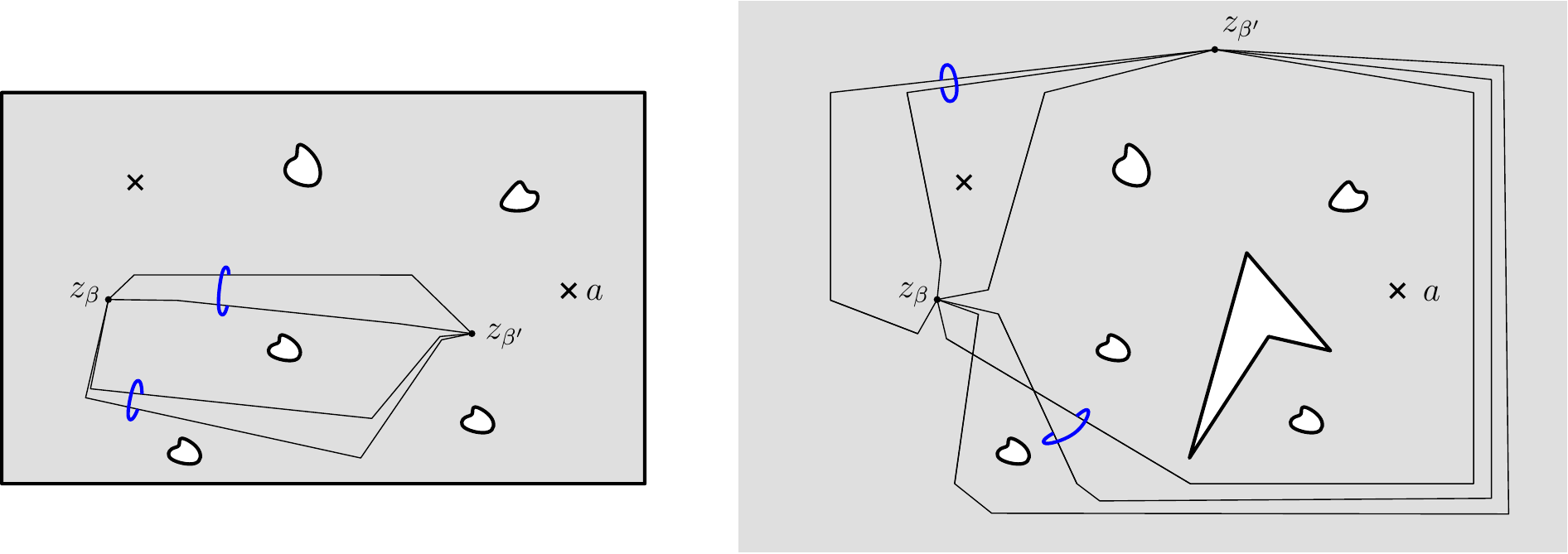}
\caption{Some of the curves $\gamma_s$ arising from Fig.~\ref{fig:polygon}, after a small perturbation, and the resulting clusters. In the left case, $\beta$ and $\beta'$ are boundaries of holes, while in the right case $\beta'$ is the exterior boundary.}
\label{fig:polygon2}
\end{figure}

\begin{lemma}\label{le:clusters1}
$S_{\beta,\beta'}$ is partitioned into ${\mathcal O}(h^2)$ clusters. Such partition can be computed in ${\mathcal O}(h n \log n)$ time.
\end{lemma}
\begin{proof}
Let $\Gamma_{\beta,\beta'}$ be the set of curves $\gamma_s$ over all segments $s\in S_{\beta,\beta'}$.
Note that two curves $\gamma_s$ and $\gamma_{s'}$ of $\Gamma_{\beta,\beta'}$ may cross only once, and they do so along $s$ and $s'$.
With a small perturbation of the curves in $\Gamma_{\beta,\beta'}$ we may assume that $\gamma_s$ and $\gamma_{s'}$ are either disjoint or cross at $s\cap s'$.
(We do not actually use that $\gamma_s$ contains shortest paths inside $C_\beta$ and $C_{\beta'}$ besides for this property of non-crossing curves inside $C_\beta$ and $C_{\beta'}$.)

We now describe a simple criteria using crossing sequences to decide when two segments of $S_{\beta,\beta'}$ are $a$-$b$ equivalent. We take a set $\Sigma$ of non-crossing paths in $X_{\beta,\beta'}$ that have the following property: cutting $X_{\beta,\beta'}$ along the curves of $\sigma$ removes all holes. Such set $\Sigma$ has a tree-like structure and can be constructed as follows. For each boundary $\alpha$ of $P$, distinct from $\beta$ and $\beta'$, we add to $\Sigma$ the shortest path in $P$ between $a$ and $\alpha$. We add to $\Sigma$ the shortest path in $P$ between $a$ and $b$. Finally, if $\beta$ or $\beta'$ is the exterior boundary of $P$, we add to $\Sigma$ a shortest path from $a$ to a point that is very far in $P$ union the the outer face. In total, $\Sigma$ has $O(h)$ polygonal paths in $X_{\beta,\beta'}$. Note that the curves in $\Sigma$ are non-crossing and a small perturbation makes them disjoint, except at the common endpoint $a$. See Fig.~\ref{fig:polygon3}. Each curve $\sigma\in \Sigma$ is simple and has two sides. We arbitrarily choose one of them as the right side and the other as the left side. We use $\sigma_1,\dots,\sigma_k$ to denote the curves of $\Sigma$.

\begin{figure}[t]
\centering
\includegraphics[width=0.9\textwidth]{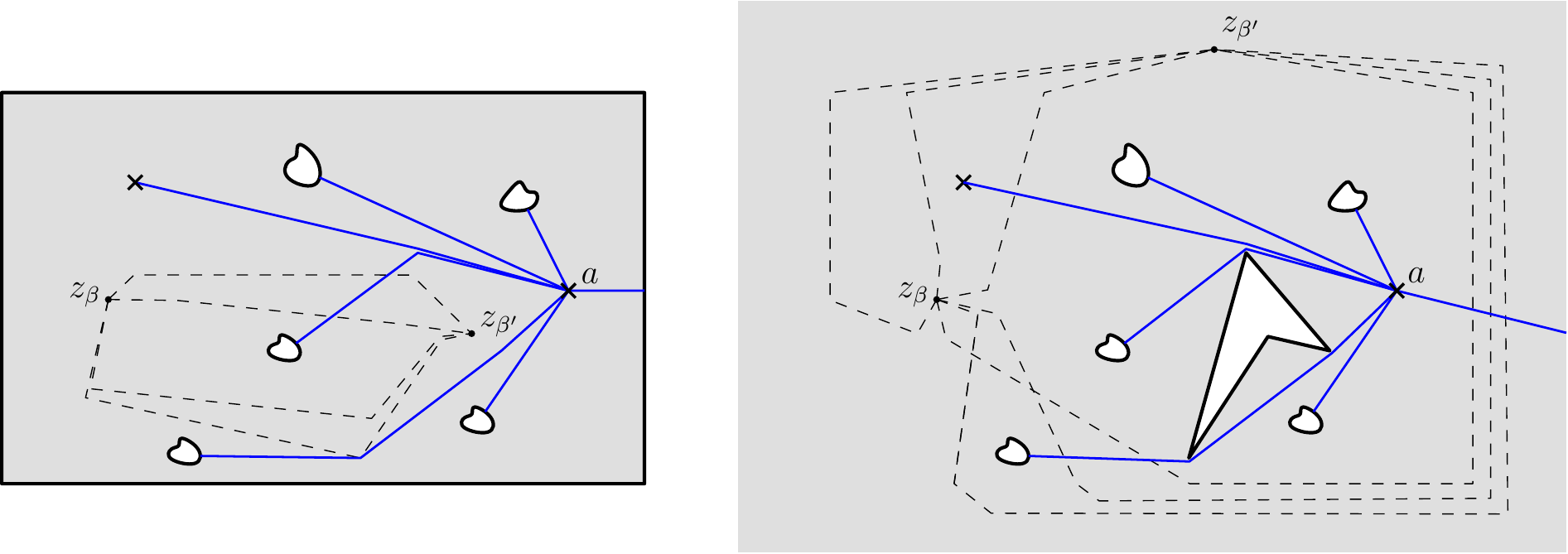}
\caption{The curves $\Sigma$ in solid and $\Gamma_{\beta,\beta'}$ in dashed style for the example of Fig.~\ref{fig:polygon2}, after a small perturbation.}
\label{fig:polygon3}
\end{figure}

To each path $\gamma$ in $\Gamma_{\beta,\beta'}$ we associate a crossing sequence $w(\gamma)$ as follows. We start with the empty word and walk along $\gamma$. When $\gamma$ crosses an arc $\sigma_i\in \Sigma$ from left-to-right we append $\sigma_i^\rightarrow$ to the word, and when $\gamma$ crosses $\sigma_i$ from right-to-left we append  $\sigma_i^\leftarrow$ to the word. From the crossing sequence $w(\gamma)$ we can obtain the \emph{reduced crossing sequence}  $w^R(\gamma)$: we iteratively remove contiguous appearances of $\sigma_i^\rightarrow$ and $\sigma_i^\leftarrow$, for any $i$. For example, from the crossing sequence $\sigma_1^\rightarrow \sigma_2^\leftarrow \sigma_3^\rightarrow \sigma_3^\leftarrow \sigma_2^\rightarrow$ we obtain the reduced crossing sequence $\sigma_1^\rightarrow$. A consequence of using $\{ \sigma_i \}$ to construct the so-called universal cover is the following characterization:
the curves $\gamma_s$ and $\gamma_{s'}$ are homotopic in $X_{\beta,\beta'}$ if and only if the curves $\gamma_s$ and $\gamma_{s'}$ have the same reduced crossing sequence. See for example~\cite{clms-04}. We conclude that $s$ and $s'$ from $S_{\beta,\beta'}$ are $a$-$b$ equivalent if and only if 
$w^R(\gamma_s)= w^R(\gamma_{s'})$.

The union of $\Sigma$ and $\Gamma_{\beta,\beta'}$ forms a family of pseudosegments: any two of them crosses at most once.
Indeed, by construction different curves can only cross in $P$, but inside $P$ all those curves are shortest paths, and thus can cross at most once.
Furthermore, the segments $\Sigma$ do not cross by construction and the curves of $\Gamma_{\beta,\beta'}$ have common endpoints.
Mount~\cite[Theorem 1.1]{mount-90} has shown that in such case the curves in $\Gamma_{\beta,\beta'}$ 
define at most ${\mathcal O}(|\Sigma|^2)={\mathcal O}(h^2)$ distinct crossing sequences. Therefore, there are at most ${\mathcal O}(h^2)$ homotopy classes defined
by the curves in $\Gamma_{\beta,\beta'}$, and $S_{\beta,\beta'}$ defines ${\mathcal O}(h^2)$ clusters. 

The procedure we have described is constructive: we have to compute $O(h)$ shortest paths in $P$ to obtain the curves of $\Sigma$, and then, for each segment $s\in S_{\beta,\beta'}$, we have to compute the corresponding crossing sequence. Such crossing sequence is already reduced. Note that for computing the crossing sequence of $\gamma_s$ we never have to construct $\gamma_s$ itself because all crossings occur along $s$.
This can be done in ${\mathcal O}(h n\log n)$ time using algorithms for shortest paths in polygonal domains~\cite{hs-oaesp-99} and data structures for ray-shooting among the segments of $\Sigma$~\cite{ray-shooting}.
\end{proof}

\begin{lemma}\label{le:clusters2}
For each cluster, either all or none of the 
segments in the cluster are crossed by a minimum-cost $a$-$b$ path. 
\end{lemma}
\begin{proof}

\begin{figure}[t]
\centering
\includegraphics[width=0.9\textwidth]{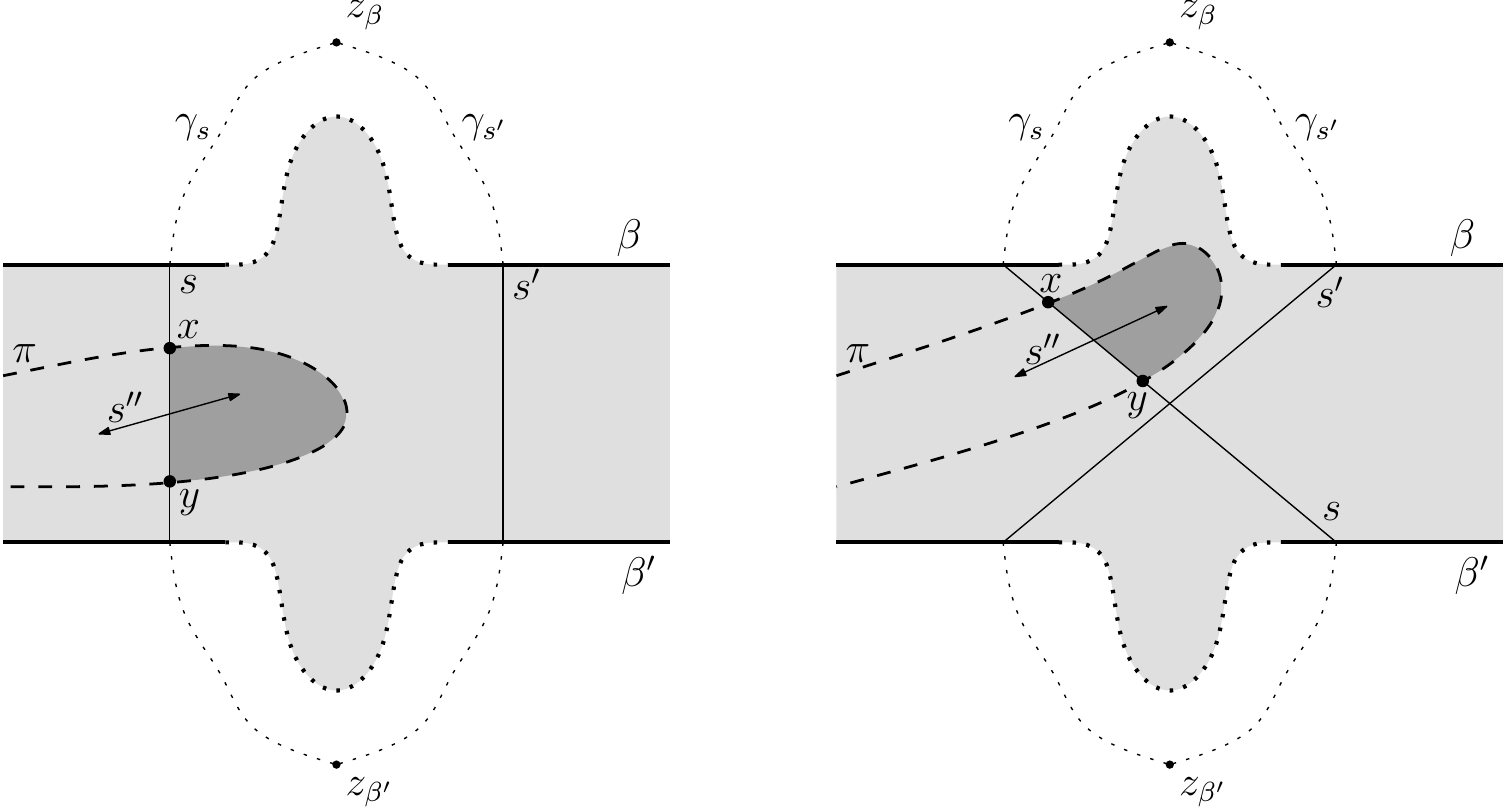}
\caption{Figure for the proof of Lemma~\ref{le:clusters2}. Left: case when $s$ and $s'$ are disjoint.
	Right: case when $s$ and $s'$ intersect. In both cases, the darker gray region represents the topological disk defined by $\pi[x,y]$ and $\segment{xy}$.}
\label{fig:polygon4}
\end{figure}

Let $s$ and $s'$ be two $a$-$b$ equivalent segments from $S_{\beta,\beta'}$.
This implies that $\gamma_s$ and $\gamma_{s'}$ are homotopic in $X_{\beta,\beta'}$.
Therefore, the path $\gamma$ obtained by concatenating $\gamma_s$ and the reversal of $\gamma_{s'}$ is contractible in $X_{\beta,\beta'}$.

Let $\pi$ be a minimum-cost path between $a$ and $b$ that crosses $s$ but does not cross $s'$.
We will reach a contradiction.
We take $\pi$ that minimizes the total number of crossings with $s$. 
We may assume that $\pi$ is simple and disjoint from $\beta,\beta'$.
We use $\pi[x,y]$ to denote the subpath of $\pi$ between points $x$ and $y$ of $\pi$. 
We distinguish two cases:
\begin{itemize}
	\item $s$ and $s'$ do not intersect. In this case, the curve $\gamma$ is simple and contractible
		in $X_{\beta,\beta'}$. It follows that $\gamma$ bounds a topological disk $D_\gamma$ in $X_{\beta,\beta'}$.
		By hypothesis, $\pi$ crosses the part of the boundary of $D_\gamma$ defined by $s$ but
		not $s'$. Therefore, $\pi$ must cross at least twice along $s$. Let $x$ and $y$ be two consecutive crossings of $\pi$ and $s$
		as we walk along $\pi$. See Fig.~\ref{fig:polygon4} left. Consider the path $\pi'$ that replaces $\pi[x,y]$ by the segment $\segment{xy}$.
		Any segment $s''$ crossing $s$ along $\segment{xy}$ crosses also $\pi$ because $\pi[x,y]$ and $\segment{xy}$ define a disk.
		Therefore $\pi'$ crosses no more segments than $\pi$ and crosses $s$ twice less than $\pi$.
		Thus, we reach a contradiction. (If $\pi'$ is not simple we can take a simple path contained in $\pi'$.)
	\item $s$ and $s'$ intersect. In this case, the curve $\gamma$ 
		in $X_{\beta,\beta'}$ has precisely one crossing. Let $\gamma'$ and $\gamma''$ be the two simple loops
		obtained by splitting $\gamma$ at its unique crossing. It must be that $\gamma'$ and $\gamma''$ are contractible,
		as otherwise $\gamma$ would not be contractible. See Fig.~\ref{fig:polygon4} right.
		Therefore, we obtain two topological disks $D_{\gamma'}$ and $D_{\gamma''}$, one bounded by $\gamma'$ and another by $\gamma''$.
		The path $\pi$ must cross the boundary of $D_{\gamma'}$ or $D_{\gamma''}$, and the same argument than in the previous item
		leads to a contradiction.\qedhere
\end{itemize}
\end{proof}

A minimum-cost $a$-$b$ path can now be found by testing all possible cluster subsets, that is, $2^{{\mathcal O}(h^4)}$ possibilities.

\begin{theorem}
The restricted {\sc $2$-Cells-Connection} problem in a polygon with $h$ holes and $n$ segments can be found in $2^{{\mathcal O}(h^4)}\polylog n$ time.
\end{theorem}
\begin{proof}
We classify the segments of $S$ into ${\mathcal O}(h^4)$ clusters using Lemma~\ref{le:clusters1}. This takes ${\mathcal O}(h^3 n\log n)$ time. Because 
of Lemma~\ref{le:clusters2}, we know that either all or none of the segments in a cluster are crossed by an optimal $a$-$b$ path. Each subset of the clusters defines 
a set of segments $S'$, and we can test whether $S'$ separates $a$ and $b$ in ${\mathcal O}(n\polylog n)$ time~\cite{single-face1,single-face2}. 
\end{proof}

\section{Connecting all cells}\label{allcells}

We show that {\sc All-Cells-Connection} is NP-hard by a reduction from the NP-hard problem of feedback vertex set (FVS) in planar graphs (c.f.~\cite{V01}): Given a planar graph $G$, find a minimum-size set of vertices $X$ such that $G-X$ is acyclic. 

First, we subdivide every edge of $G$ obtaining a planar bipartite graph $G'$. It is clear that $G'$ has a feedback vertex set of size $k$ if and only if $G$ has one.  Next, we use the result by de Fraysseix et al.~\cite{dFMP91} (see also Hartman et al.~\cite{HNZ91}), 
which states that every planar bipartite graph is the intersection graph of horizontal and vertical segments, where no two of them cross (intersect at a common interior point). Let $S$ be the set of segments whose intersection graph is $G'$; it can be constructed in polynomial time. Since $G'$ has no triangles, no three segments of $S$ intersect at a point. Then, observe that all cells in $\mathcal{A}(S)$ become connected by removing $k$ segments if and only if $G'$ has a feedback vertex set of size $k$. Therefore we have:

\begin{theorem}
{\sc All-Cells-Connection} in NP-hard even if no three segments intersect at a point and there are no segment crossings.
\end{theorem}

It is also easy to see that if no three segments intersect at a point a $k$-size solution to {\sc All-Cells-Connection}
corresponds to a $k$-size solution of FVS in the intersection graph of the input segments. 
For general graphs, FVS is fixed-parameter tractable when parameterized with the size of the solution~\cite{CFLLV08}, and has a polynomial-time $2$-approximation algorithm~\cite{V01}. We thus obtain the following:

\begin{corollary}
When no three segments intersect at a point, {\sc All-Cells-Connection} is fixed-parameter tractable with respect to the size of the solution and has a polynomial-time $2$-approximation algorithm.
\end{corollary}


\section*{Acknowledgments:}
We would like to thank Primo{\v z} {\v S}kraba
 for bringing to our attention some of the problems studied in this abstract.

\bibliographystyle{alpha}
\bibliography{segments_full}

\newcommand{\etalchar}[1]{$^{#1}$}
\begin{thebibliography}{KKMO07}

\bibitem[ACGK11]{acgk-11}
H.~Alt, S.~Cabello, P.~Giannopoulos, and C.~Knauer.
\newblock On some connection problems in straight-line segment arrangements.
\newblock In {\em Abstracts of the 27th European Workshop on Computational
  Geometry (EuroCG 11)}, pages 27--30, 2011.

\bibitem[BFC04]{lca1}
M.~A. Bender and M.~Farach-Colton.
\newblock The level ancestor problem simplified.
\newblock {\em Theor. Comput. Sci.}, 321(1):5--12, 2004.

\bibitem[BK09]{BK09}
S.~Bereg and D.~G. Kirkpatrick.
\newblock Approximating barrier resilience in wireless sensor networks.
\newblock In {\em Proc. 5th ALGOSENSORS}, volume 5804 of {\em LNCS}, pages
  29--40. Springer, 2009.

\bibitem[BLWZ05]{BLWZ05}
H.~Broersma, X.~Li, G.~Woeginger, and S.~Zhang.
\newblock Paths and cycles in colored graphs, p.299.
\newblock {\em Australasian J. Combin.}, 31:299--312, 2005.

\bibitem[CdVL10]{ccl-fsntc-10}
S.~Cabello, {\'E}~Colin de~Verdi{\`e}re, and F.~Lazarus.
\newblock Finding shortest non-trivial cycles in directed graphs on surfaces.
\newblock In {\em Proc. 26th ACM SoCG}, pages 156--165, 2010.

\bibitem[CEG{\etalchar{+}}94]{ray-shooting}
B.~Chazelle, H.~Edelsbrunner, M.~Grigni, L.~J. Guibas, J.~Hershberger,
  M.~Sharir, and J.~Snoeyink.
\newblock Ray shooting in polygons using geodesic triangulations.
\newblock {\em Algorithmica}, 12(1):54--68, 1994.

\bibitem[CFL{\etalchar{+}}08]{CFLLV08}
J.~Chen, F.~V. Fomin, Y.~Liu, S.~Lu, and Y.~Villanger.
\newblock Improved algorithms for feedback vertex set problems.
\newblock {\em J. Comput. Syst. Sci.}, 74:1188--1198, 2008.

\bibitem[CLMS04]{clms-04}
S.~Cabello, Y.~Liu, A.~Mantler, and J.~Snoeyink.
\newblock Testing homotopy for paths in the plane.
\newblock {\em Discr. {\&} Comp. Geometry}, 31(1):61--81, 2004.

\bibitem[dBDS95]{single-face2}
M.~de~Berg, K.~Dobrindt, and O.~Schwarzkopf.
\newblock On lazy randomized incremental construction.
\newblock {\em Discr. {\&} Comp. Geometry}, 14(3):261--286, 1995.

\bibitem[dFOP91]{dFMP91}
H.~de~Fraysseix, P.~{Ossona de Mendez}, and J.~Pach.
\newblock Representation of planar graphs by segments.
\newblock {\em Intuitive Geometry}, 63:109--117, 1991.

\bibitem[FGI10]{FGK10}
M.~R. Fellows, J.~Guo, and I.~Iyad.
\newblock The parameterized complexity of some minimum label problems.
\newblock {\em J. Comput. Syst. Sci.}, 76:727--740, 2010.

\bibitem[GKV11]{gkv-ipud-11}
M.~Gibson, G.~Kanade, and K.~Varadarajan.
\newblock On isolating points using disks.
\newblock Manuscript available at \url{http://arxiv.org/abs/1104.5043}, 2011.

\bibitem[GSS89]{single-face1}
L.~J. Guibas, M.~Sharir, and S.~Sifrony.
\newblock On the general motion-planning problem with two degrees of freedom.
\newblock {\em Discr. {\&} Comp. Geometry}, 4:491--521, 1989.

\bibitem[H{\aa}s01]{Hastad01}
J.~H{\aa}stad.
\newblock Some optimal inapproximability results.
\newblock {\em J. ACM}, 48(4):798--859, 2001.

\bibitem[HMS07]{HMS07}
R.~Hassin, J.~Monnot, and D.~Segev.
\newblock Approximation algorithms and hardness results for labeled
  connectivity problems.
\newblock {\em J. Comb. Optim.}, 14(4):437--453, 2007.

\bibitem[HNZ91]{HNZ91}
I.~B. Hartman, I.~Newman, and R.~Ziv.
\newblock On grid intersection graphs.
\newblock {\em Discrete Mathematics}, 87:41--52, 1991.

\bibitem[HS99]{hs-oaesp-99}
J.~Hershberger and S.~Suri.
\newblock An optimal algorithm for euclidean shortest paths in the plane.
\newblock {\em SIAM J. Comput.}, 28(6):2215--2256, 1999.

\bibitem[HT84]{lca2}
D.~Harel and R.~E. Tarjan.
\newblock Fast algorithms for finding nearest common ancestors.
\newblock {\em SIAM J. Comput.}, 13(2):338--355, 1984.

\bibitem[KH07]{KH07}
S.~Kloder and S.~Hutchinson.
\newblock Barrier coverage for variable bounded-range line-of-sight guards.
\newblock In {\em Proc. ICRA}, pages 391--396. IEEE, 2007.

\bibitem[KKMO07]{kkmo-oir-07}
S.~Khot, G.~Kindler, E.~Mossel, and R.~O'Donnell.
\newblock Optimal inapproximability results for max-cut and other 2-variable
  csps?
\newblock {\em SIAM J. Comput.}, 37(1):319--357, 2007.

\bibitem[KLA07]{KLA07}
S.~Kumar, T.-H. Lai, and A.~Arora.
\newblock Barrier coverage with wireless sensors.
\newblock {\em Wireless Networks}, 13(6):817--834, 2007.

\bibitem[Mou90]{mount-90}
David~M. Mount.
\newblock The number of shortest paths on the surface of a polyhedron.
\newblock {\em SIAM J. Comput.}, 19(4):593--611, 1990.

\bibitem[MT01]{mt-gs-01}
B.~Mohar and C.~Thomassen.
\newblock {\em Graphs on surfaces}.
\newblock Johns Hopkins Studies in the Mathematical Sciences. John Hopkins
  University Press, 2001.

\bibitem[Tho90]{t-egsnc-90}
C.~Thomassen.
\newblock Embeddings of graphs with no short noncontractible cycles.
\newblock {\em J. of Comb. Theory, B}, 48(2):155--177, 1990.

\bibitem[Tse11]{tseng-thesis}
K.-C.~R. Tseng.
\newblock Resilience of wireless sensor networks.
\newblock Master's thesis, The University Of British Columbia (Vancouver),
  April 2011.

\bibitem[Vaz01]{V01}
V.~V. Vazirani.
\newblock {\em Approximation Algorithms}.
\newblock Springer, 2001.

\end{thebibliography}

\end{document}